\documentclass[10 pt]{article}
\usepackage{graphicx} 
\usepackage[a4paper, total={6.5in, 9in}]{geometry} 
\usepackage{amsfonts}
\usepackage{amsthm}
\usepackage{amsmath}
\usepackage[T1]{fontenc}
\usepackage{stmaryrd} 
\usepackage[sortcites = true, backend=biber, style=numeric, url = false, maxbibnames=10,maxcitenames=10]{biblatex} 
\usepackage[shortlabels]{enumitem} 
\usepackage{hyperref}
\hypersetup{
    colorlinks=true,
    linkcolor=blue,
    filecolor=magenta,      
    urlcolor=cyan,
    pdftitle={Graded Dirac and graded Poisson structures},
    pdfpagemode=FullScreen,
    }
\usepackage[capitalize, nameinlink]{cleveref} 

\crefformat{subsection}{#2Subsection~#1#3}

\usepackage{xcolor} 
\usepackage{quiver} 
\usepackage{orcidlink}
\usepackage{authblk} 
\usepackage{physics} 
\usepackage{tikz}

\usepackage{titlesec} 
\titleformat{\section}
  {\normalfont\large\bfseries}{\thesection.}{1em}{}
\titleformat{\subsection}
  {\normalfont\large\itshape}{\thesubsection.}{1em}{}
\titleformat{\subsubsection}
  {\normalfont\normalsize\itshape}{\thesubsubsection.}{1em}{}

\usepackage{commands}

\newenvironment{keywords}
{
\begin{center}
\textbf{Keywords}\\
\vspace{0.17cm}
\begin{minipage}{14.5cm}}
{\footnotesize
\end{minipage}
\end{center}}

\newenvironment{Abstract}
{
\begin{center}
\textbf{Abstract}\\
\vspace{0.25cm}
\begin{minipage}{14.5cm}}
{\footnotesize
\end{minipage}
\end{center}}

\newcommand{\mail}[1]{\small\href{mailto:#1}{#1}}

\newtheorem{theorem}{Theorem}[subsection]
\newtheorem{proposition}{Proposition}[subsection]
\newtheorem{lemma}{Lemma}[subsection]
\newtheorem{corollary}{Corollary}[subsection]
\newtheorem{example}{Example}[subsection]

\theoremstyle{definition}
\newtheorem{Def}{Definition}[subsection]
\newtheorem{remark}{Remark}[subsection]
\newtheorem{obs}{Observation}[subsection]

\begin{document}
\title{\huge Graded Poisson and Graded Dirac structures}

\author{Manuel de León \orcidlink{0000-0002-8028-2348} \\ \mail{mdeleon@icmat.es}}
\affil{Instituto de Ciencias Matemáticas, Campus Cantoblanco, Consejo Superior de Investigaciones Científicas, C/Nicolás Cabrera, 13–15, Madrid 28049, Spain}
\affil{Real Academia de Ciencias Exactas, Físicas y Naturales de España, C/Valverde, 22, Madrid 28004, Spain}

\author{ Rubén Izquierdo-López \orcidlink{0009-0007-8747-344X} \\ \mail{rubizqui@ucm.es}}
\affil{Instituto de Ciencias Matemáticas, Campus Cantoblanco, Consejo Superior de Investigaciones Científicas, C/Nicolás Cabrera, 13–15, Madrid 28049, Spain}

\date{\today}

\maketitle

\begin{Abstract}
There have been several attempts in recent years to extend the notions of symplectic and Poisson structures in order to create a suitable geometrical framework for classical field theories, trying to achieve a success similar to the use of these concepts in Hamiltonian mechanics. These notions always have a graded character, since the multisymplectic forms are of a higher degree than two. Another line of work has been to extend the concept of Dirac structures to these new scenarios. In the present paper we review all these notions, relate them and propose and study a generalization that (under some mild regularity conditions) includes them and is of graded nature. We expect this generalization to allow us to advance in the study of classical field theories, their integrability, reduction, numerical approximations and even their quantization.
\end{Abstract}


\tableofcontents

\section{Introduction}

As it is well-known, mechanics experienced a drastic change as soon as it was able to use symplectic geometry in its description. This occurred in the 50's and 60's of the last century, and made it possible to obtain Hamilton's equations as the integral curves of a vector field on a symplectic manifold, in fact, on the cotangent bundle of the configuration space of the system. This implied a liberation from coordinates and the possibility of obtaining the usual properties of mechanical systems (conservation of energy, other conserved quantities, Noether's theorem, integration, reduction procedures, Hamilton-Jacobi theory, Arnold--Liouville theorem and corresponding action-angle coordinates\ldots) in a simple and elegant way.\\

Consequently, one direction of research has been to extend symplectic geometry to more general situations describing classical field theories. At the end of the 60's, three groups of physicist-mathematicians independently developed a formalism called multisymplectic, which sought to extend the symplectic of mechanics to this case \cite{Goldschmidt1973,Garcia1973,Kijowski1979}. The difficulty of this new geometry is that, while symplectic geometry is very rigid (it is always locally equivalent to the canonical one of a cotangent bundle via the Darboux theorem), the situation is very different in the multisymplectic case. In the last 50 years, much effort has been made to achieve progress in this leading, and despite many achievements, a theory as satisfactory as for mechanics has not yet been achieved, being still a field of research in full development. 
\\

A symplectic structure $\omega$ on a manifold $M$ determines an algebraic structure in the function algebra $C^\infty(M)$ through the Poisson bracket defined by 
$$
\{f, g\} = \omega(X_f, X_g),
$$
where $f, g \in C^\infty(M)$ and $X_f$ and $X_g$ are the corresponding Hamiltonian vector fields. The existence of the Poisson bracket allows not only to express the evolution of the observables, but it is also key for quantization processes. But this algebraic structure on $C^\infty(M)$ can also be interpreted on the manifold itself by defining the Poisson tensor as the bivector $\Lambda$ given by
$$
\Lambda(df, dg) = \{f, g\},
$$
see \cite{Lichnerowicz1977}.
It is therefore natural that research has been directed towards generalizations of these notions in order to apply them to classical field theories. The brackets were extended in \cite{Cantrijn1996} (see also \cite{Leon2024a}, and more recently, \cite{GayBalmaz2024}), defining a graded algebra. In this case, the duality of Hamiltonian vector fields versus 1-forms is naturally extended to multivector fields and higher degree differential forms, since the multisymplectic form allows for several levels.\\ 

Another piece in our way is the notion of Dirac structures: They were introduced independently by T. Courant \cite{Courant1986, Courant1990} and I. Dorfman \cite{Dorfman1987, dorfman} as a simultaneous generalization of presymplectic and Poisson structures. A Dirac structure on a manifold $M$ is a maximal isotropic and involutive vector subbundle $D$ of $TM \oplus T^*M$. These structures have remarkable properties that have been used in differential geometry as well as in classical mechanics. There has been several approaches to generalize Dirac geometry to classical field theories.\\

In \cite{Vankerschaver2012}, looking for a formalism which unifies both the Lagrangian and Hamiltonian setting (usually called Skinner-Rusk formalism, see \cite{Skinner1983}), J. Vankerschaver, H. Yoshimura and M. Leok defined a multi-Dirac structure of degree $k$ on a manifold $M$ to be a sequence of vector sub-bundles $D_1, \dots, D_k$,
$$
D_p \subseteq E_p = \bigvee_p M \oplus_M \bigwedge^{k+1-p} M,
$$ 
where $1 \leq p \leq k,$ and $k$ is a fixed integer,
which will be the degree of the multi-Dirac structure. This sequence of multivector fields and forms satisfies a certain maximally isotropic property, together with involutivity with respect to the Courant bracket, defined as 
$$\llbracket (U, \alpha), (V, \beta) \rrbracket := \left([U, V], (-1)^{(p-1)q} \pounds_U \beta +
         (-1)^{q} \pounds_V \alpha - (-1)^q \frac{1}{2} d( \iota_V \alpha  + (-1)^{pq} \iota_U \beta)\right).$$
The case $k=1$ recovers the usual notion of Dirac structure in mechanics.  In this paper, the authors show
that there exists a graded multiplication and a graded bracket on the space of sections of $D_1, \dots, D_k,$ and that the latter is endowed with the structure of a Gerstenhaber algebra with respect to these two operations. Furthermore,
they define a multi-Poisson bracket on a distinguished subset of the space of forms and show that this bracket satisfies the graded Jacobi identity up to exact forms.\\

Later, in \cite{Zambon2012}, M. Zambon noticed that given a multi-Dirac structure of order $k$ on a manifold $M$, $D_1, \dots, D_k$, then $D_1$
determines completely the multi-Dirac structure by the equality $D_p = (D_1)^{\perp, p},$ for all $1 \leq p \leq k,$ where $$(D_p)^{\perp, q} = \{(u, \alpha) \in E_q : \iota_u \beta - (-1)^{pq} \iota_v \alpha = 0, \forall (v, \beta) \in D_q\}.$$ 
He then introduced
the notion of higher Dirac structure of order $k$ on a manifold $M$ as a vector subbundle
   
\[\begin{tikzcd}
	D & {E_1 = TM \oplus_M \bigwedge^{k}M} \\
	M
	\arrow[hook, from=1-1, to=1-2]
	\arrow[from=1-1, to=2-1]
	\arrow[curve={height=-12pt}, from=1-2, to=2-1]
\end{tikzcd}\]
satisfying the following properties
\begin{enumerate}[i)]
    \item It is \textit{Lagrangian}, that is $D = D ^{\perp, 1}.$
    \item It is \textit{involutive} with respect to the Courant bracket.
\end{enumerate}
M. Zambon also proved that, under some regularity conditions, both notions coincide.\\

Regarding the generalization of Poisson geometry, in \cite{Bursztyn2011}, H. Bursztyn proposed a definition of higher-Poisson structure that generalizes the usual one.

\begin{Def}[Higher Poisson structure]
    A \textbf{higher Poisson structure} of order $k$ on a manifold $M$ is a pair $(S,\sharp)$, where
    $S$ is a vector subbundle 
    \[\begin{tikzcd}
        S & {\bigwedge^kM} \\
        M
        \arrow[hook, from=1-1, to=1-2]
        \arrow[from=1-1, to=2-1]
        \arrow[curve={height=-12pt}, from=1-2, to=2-1]
    \end{tikzcd},\]
    and $\sharp$ is a vector bundle mapping $$\sharp: S \rightarrow TM$$ satisfying:
    \begin{enumerate}[i)]
        \item $S^{\circ, 1} = 0$, where $S^{\circ, 1}$ denotes the first annihilator. In general, for a family of forms $S \subseteq \bigwedge^a M$,
        and $p \leq a$, we define $$S^{\circ, p} := \{U \in \bigvee_p M:\,\, \iota_U \alpha = 0, \forall \alpha \in S\}.$$
        \item $\sharp$ is \textit{skew-symmetric}, that is, $\iota_{\sharp(\alpha)} \beta = - \iota_{\sharp(\beta)} \alpha,$ for all $\alpha, \beta \in S.$
        \item $S$ is \textit{involutive} with respect to the bracket\footnote{Although this notation overlaps with Schouten-Nijenhuis bracket, the meaning will be clear from the context.} 
        $$[\alpha, \beta] := \pounds_{\sharp(\alpha)} \beta - \iota_{\sharp(\alpha)} \d{\beta}$$
        and satisfies $$\sharp([\alpha, \beta]) = [\sharp(\alpha), \sharp(\beta)].$$
    \end{enumerate} 
\end{Def}
For $k = 1$ we recover the classical notion of a Poisson manifold.\\

However, in \cite{Bursztyn2019}, H. Bursztyn, N. Martinez-Alba and R. Rubio observed that, in general, the graph of a higher Poisson structure 
$$D := \{(\sharp(\alpha), \alpha): \,\, \alpha \in S\} \subseteq E_1 = TM \oplus_M \bigwedge^k M$$
\textit{does not} define a higher Dirac structure. Nevertheless, they noticed that it does define a \textit{weak higher Dirac structure} (what they
simply call higher Dirac), a not so restrictive version of the definition introduced by M. Zambon that allows for higher Poisson structures to be included.\\

The above definitions and results introduce new notions to be explored in the field of geometries underlying classical field theories.\\

The objectives of this paper are twofold. First, to integrate all these geometric and algebraic notions into a common framework that allows them to be related and extended. Second, to clarify the necessary notions:

\begin{enumerate}[i)]
    
\item A definition of higher order Poisson structure in terms of graded multivectors that extends the usual one, which we call graded Poisson.

\item A definition of a graded algebraic Poisson bracket that is in bijective correspondence with the geometric definition, obtaining the characterization of dynamics in terms of the observable algebra;

\item A notion of a graded Dirac structure that integrates the two previous concepts (at least under some mild assumptions), as well as the multisymplectic structures (which we do not ask to fulfill any regularity condition, i.e., they are only closed forms of higher degree).
\end{enumerate}

The following diagram summarizes how all of these concepts 
relate to each other:
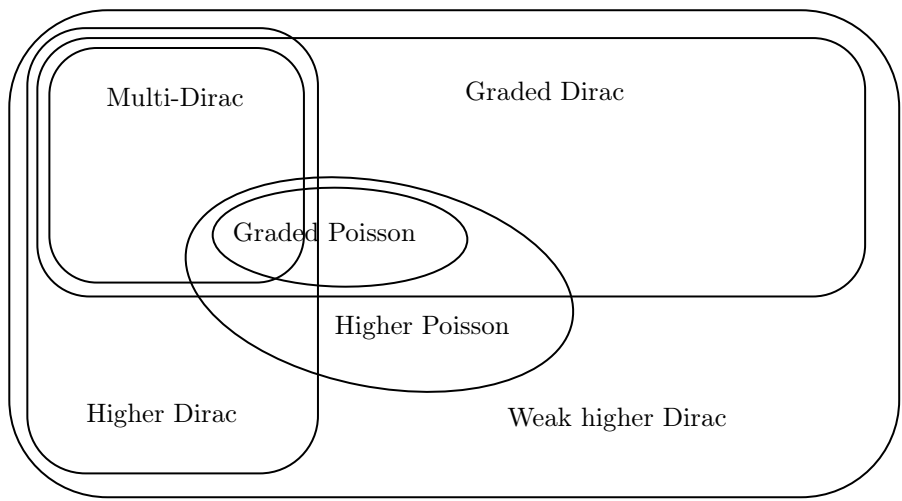
\begin{center}

\tikzset{every picture/.style={line width=0.75pt}} 

\begin{tikzpicture}[x=0.75pt,y=0.75pt,yscale=-1,xscale=1]
    
    \draw  (135,68.6) .. controls (135,41.54) and (156.94,19.6) .. (184,19.6) -- (530,19.6) .. controls (557.06,19.6) and (579,41.54) .. (579,68.6) -- (579,215.6) .. controls (579,242.66) and (557.06,264.6) .. (530,264.6) -- (184,264.6) .. controls (156.94,264.6) and (135,242.66) .. (135,215.6) -- cycle ;
    \draw  (149,59.6) .. controls (149,45.24) and (160.64,33.6) .. (175,33.6) -- (536,33.6) .. controls (550.36,33.6) and (562,45.24) .. (562,59.6) -- (562,137.6) .. controls (562,151.96) and (550.36,163.6) .. (536,163.6) -- (175,163.6) .. controls (160.64,163.6) and (149,151.96) .. (149,137.6) -- cycle ;
    \draw  (144,57.6) .. controls (144,41.58) and (156.98,28.6) .. (173,28.6) -- (260,28.6) .. controls (276.02,28.6) and (289,41.58) .. (289,57.6) -- (289,223.6) .. controls (289,239.62) and (276.02,252.6) .. (260,252.6) -- (173,252.6) .. controls (156.98,252.6) and (144,239.62) .. (144,223.6) -- cycle ;
    \draw (155,62.2) .. controls (155,49.17) and (165.57,38.6) .. (178.6,38.6) -- (258.4,38.6) .. controls (271.43,38.6) and (282,49.17) .. (282,62.2) -- (282,133) .. controls (282,146.03) and (271.43,156.6) .. (258.4,156.6) -- (178.6,156.6) .. controls (165.57,156.6) and (155,146.03) .. (155,133) -- cycle ;
    \draw (223.51,138.65) .. controls (229.01,110.74) and (276.52,96.6) .. (329.62,107.06) .. controls (382.72,117.53) and (421.31,148.64) .. (415.81,176.55) .. controls (410.31,204.46) and (362.8,218.6) .. (309.7,208.13) .. controls (256.6,197.67) and (218.01,166.56) .. (223.51,138.65) -- cycle ;
    \draw(236.48,132.46) .. controls (236.77,118.73) and (265.44,108.19) .. (300.52,108.92) .. controls (335.61,109.65) and (363.82,121.37) .. (363.53,135.09) .. controls (363.25,148.82) and (334.58,159.36) .. (299.49,158.63) .. controls (264.41,157.9) and (236.2,146.18) .. (236.48,132.46) -- cycle ;
    
    \draw (382,218) node [anchor=north west][inner sep=0.75pt]  [xscale=1,yscale=1] [align=left] {Weak higher Dirac};
    \draw (361,54) node [anchor=north west][inner sep=0.75pt]  [xscale=1,yscale=1] [align=left] {Graded Dirac};
    \draw (182,57) node [anchor=north west][inner sep=0.75pt]  [xscale=1,yscale=1] [align=left] {Multi-Dirac};
    \draw (172,216) node [anchor=north west][inner sep=0.75pt]  [xscale=1,yscale=1] [align=left] {Higher Dirac};
    \draw (296,172) node [anchor=north west][inner sep=0.75pt]  [xscale=1,yscale=1] [align=left] {Higher Poisson};
    \draw (245,125) node [anchor=north west][inner sep=0.75pt]  [xscale=1,yscale=1] [align=left] {Graded Poisson};

    \end{tikzpicture}

\end{center}

The extent to which graded Dirac structures (resp. graded Poisson) fail to include weak higher Dirac structures (resp. higher Poisson structures) is characterized by some mild regularity conditions, essentially proving that these two concepts coincide (see \cref{Dirac_Structures:Overview:Thm:Equivalence_HigherDirac_multiDirac}, \cref{Graded_Poisson_structures:thm:Equivalence}).\\

The main results of the paper are the following: 

\begin{enumerate}[i)]
    \item The first main result of this paper is \cref{Dirac_structures:thm:Differential_equivalence}, where we show that graded Dirac structures are characterized by the structure on degree $n$, generalizing the equivalent result in higher Dirac geometry by M. Zambon in \cite{Zambon2012}.
    \item We also study the natural foliation induced by any graded Poisson manifold (\cref{Graded_Poisson_structures:thm:Poisson_implies_Multisymplecticfoliation}), and we propose a way of recovering a graded Poisson structure from a non-degenerate multisymplectic foliation in \cref{Graded_Poisson_structures:thm:Multisymplecticfoliation_implies_Poisson}. These results are then extended to graded Dirac manifolds in \cref{Graded_Poisson_structures:thm:Multisymplecticfoliation_implies_Dirac}.
    \item The induced graded Poisson bracket by a graded Poisson structure is also studied, and we prove the last main result of this paper, \cref{Graded_Poisson_brackets:thm:Bracket_Implies_Structure}, where we show that under some integrability conditions on the family of form subbundles, any graded Poisson bracket induces a graded Poisson structure, thus providing the equivalence between the geometric and algebraic aspect of brackets, similar to the correspondence present in Poisson geometry. The equivalent result is then extended to graded Dirac structures in \cref{GradedPoissonStructures:thm:Bracket_implies_GradedDirac}.
\end{enumerate}

The diagram below shows the interplay between the introduced concepts, comparing the case in classical mechanics, and the case in classical field theories, showing the achieved generalization of the structures appearing in classical mechanics to a graded nature.

\begin{center}
\tikzset{every picture/.style={line width=0.75pt}} 

\begin{tikzpicture}[x=0.75pt,y=0.75pt,yscale=-1,xscale=1]
    
    \draw (428.77,108.85) .. controls (428.77,89.38) and (444.55,73.6) .. (464.02,73.6) -- (569.75,73.6) .. controls (589.22,73.6) and (605,89.38) .. (605,108.85) -- (605,220.35) .. controls (605,239.82) and (589.22,255.6) .. (569.75,255.6) -- (464.02,255.6) .. controls (444.55,255.6) and (428.77,239.82) .. (428.77,220.35) -- cycle ;
    \draw (463.07,208.23) .. controls (463.37,188.26) and (494.49,173.63) .. (532.58,175.53) .. controls (570.67,177.43) and (601.31,195.16) .. (601.01,215.12) .. controls (600.71,235.08) and (569.59,249.72) .. (531.49,247.81) .. controls (493.4,245.91) and (462.77,228.19) .. (463.07,208.23) -- cycle ;
    \draw (370.07,164.49) .. controls (370.28,146.43) and (413.18,133.41) .. (465.89,135.4) .. controls (518.6,137.38) and (561.17,153.63) .. (560.96,171.69) .. controls (560.75,189.75) and (517.85,202.78) .. (465.14,200.79) .. controls (412.43,198.8) and (369.87,182.55) .. (370.07,164.49) -- cycle ;
    \draw    (342,195.15) .. controls (340.88,174.72) and (362.12,161.46) .. (390.27,164.01) ;
    \draw [shift={(392,164.18)}, rotate = 186.41] [color={rgb, 255:red, 0; green, 0; blue, 0 }  ][line width=0.75]    (10.93,-3.29) .. controls (6.95,-1.4) and (3.31,-0.3) .. (0,0) .. controls (3.31,0.3) and (6.95,1.4) .. (10.93,3.29)   ;
    \draw (95.37,111.85) .. controls (95.37,92.38) and (111.15,76.6) .. (130.62,76.6) -- (236.35,76.6) .. controls (255.82,76.6) and (271.6,92.38) .. (271.6,111.85) -- (271.6,223.35) .. controls (271.6,242.82) and (255.82,258.6) .. (236.35,258.6) -- (130.62,258.6) .. controls (111.15,258.6) and (95.37,242.82) .. (95.37,223.35) -- cycle ;
    \draw  (152.92,206.04) .. controls (151.68,183.42) and (175.4,165.17) .. (205.9,165.28) .. controls (236.41,165.39) and (262.15,183.81) .. (263.39,206.43) .. controls (264.63,229.04) and (240.9,247.29) .. (210.4,247.19) .. controls (179.89,247.08) and (154.16,228.66) .. (152.92,206.04) -- cycle ;
    \draw (106.04,150.67) .. controls (105.35,131.55) and (131.84,115.66) .. (165.21,115.19) .. controls (198.58,114.72) and (226.2,129.84) .. (226.89,148.96) .. controls (227.58,168.09) and (201.08,183.97) .. (167.71,184.44) .. controls (134.34,184.92) and (106.73,169.8) .. (106.04,150.67) -- cycle ;
    
    \draw (471.36,103.63) node [anchor=north west][inner sep=0.75pt]  [xscale=1,yscale=1] [align=left] {Graded Dirac};
    \draw (481.08,207.63) node [anchor=north west][inner sep=0.75pt]  [xscale=1,yscale=1] [align=left] {Graded Poisson};
    \draw (415.64,157.63) node [anchor=north west][inner sep=0.75pt]  [xscale=1,yscale=1] [align=left] {Multisymplectic};
    \draw (280.91,197.25) node [anchor=north west][inner sep=0.75pt]  [xscale=1,yscale=1] [align=left] {Forms $\displaystyle \omega $ such that\\$\displaystyle \operatorname{Im} \flat _{p}$ does not have\\constant rank};
    \draw (127.96,91.3) node [anchor=north west][inner sep=0.75pt]  [xscale=1,yscale=1] [align=left] {Dirac};
    \draw (185.68,199.02) node [anchor=north west][inner sep=0.75pt]  [xscale=1,yscale=1] [align=left] {Poisson\\};
    \draw (117.24,137.91) node [anchor=north west][inner sep=0.75pt]  [xscale=1,yscale=1] [align=left] {Pre-symplectic\\};
    \draw (112,44) node [anchor=north west][inner sep=0.75pt]  [xscale=1,yscale=1] [align=left] {Classical mechanics};
    \draw (444,36) node [anchor=north west][inner sep=0.75pt]  [xscale=1,yscale=1] [align=left] {Classical field theories};

    \end{tikzpicture}
\end{center}

The paper is structured as follows. In \cref{section:Overview} we first give some basic definitions in multisymplectic geometry and recall some fundamental things about Dirac structures. Later, we review the different notions of multi-Dirac structures, higher Dirac structures, and weak higher Dirac structures existing in the literature in \cref{Overview:subsection:multidirac}, \cref{Overview:subsection:higher_Dirac}, and \cref{Overview:subsection:Weak_Higher_Dirac}, respectively. In \cref{section:Graded_Dirac_structures} we introduce our own definition of graded Dirac structures, beginning by studying the linear case, then extending this notion to the realm of manifolds. In this section we also study the relation of graded Dirac structures to the existing structures in the literature.
\cref{section:Graded_Poisson_structures} is devoted to discuss graded Poisson manifolds and graded Poisson brackets as well as their relations. In addition, in \cref{Graded_Poisson_structures:subsection:Currents} we apply the previous results to study currents. Finally, in \cref{section:Conclusions} we highlight the main results of the paper and give some comments about our future work. We also include an appendix to present the definition and main properties of the so-called Schouten-Nijenhuis bracket that we are constantly using
without mention along the paper.

\begin{center}{\bf Notation and conventions}
\end{center}
We will use these notations throughout the paper; to facilitate the reading, they are collected here.
\begin{enumerate}
    \item All manifolds are assumed $C^\infty$-smooth and finite dimensional.
    \item Einstein's summation convention is assumed throughout the text, unless stated otherwise.
    \item $\bigvee_p M$ denotes the vector bundle of $p$-vectors on $M$, $\bigwedge^p TM.$
    \item $\mathfrak{X}^p(M) = \Gamma\left(\bigvee_p M\right)$ denotes the space of all multivector fields of order $p$.
    \item $\bigwedge^aM$ denotes the vector bundle of $a$-forms on $M$, $\bigwedge^a T^\ast M.$
    \item $\Omega^a(M) = \Gamma\left (\bigwedge^a M\right )$ denotes the space of all $a$-forms on $M$.
    \item $[\cdot, \cdot]$ denotes the Schouten-Nijenhuis bracket on multivector fields. We use the sign conventions of \cite{Marle1997}.
    \item $\pounds_U \alpha = \d{\iota_U\alpha} - (-1)^p \iota \d{\alpha}$, for $U \in \mathfrak{X}^p(M)$,
    $\alpha \in \Omega^a(M)$ denotes the Lie derivative along multivector fields. 
    For a proof of its main properties, we refer to \cite{Forger2003a}.
    \item For a subbundle $K \subseteq \bigvee_p M,$ and for $a \geq p,$ we denote by $$K^{\circ, a} = \{\alpha \in \bigwedge^a M: \iota_K \alpha = 0\}$$ the annihilator of order $a$ of $K$.
    \item Similarly, for a subbundle $S \subseteq \bigwedge^a M,$ and $p \leq a$ we denote by 
    $$S^{\circ, p} = \{U \in \bigvee_p M: \iota_U S = 0\}$$ the annihilator of order $p$ of $S$.
\end{enumerate}


\section{An overview of previous generalizations of Dirac structures to classical field theories}
\label{section:Overview}

\subsection{Multisymplectic manifolds}
\label{Overview:subsection:multisymplectic}

In this subsection we recall some basic definitions in Multisymplectic geometry.

\begin{Def}[Multisymplectic manifold] A \textbf{multisymplectic manifold} of order $k$ is a pair $(M, \omega),$ where $M$ is a manifold, and $\omega \in \Omega^{k+1}(M)$ is a closed $(k+1)$-form. Both the multisymplectic manifold and the multisymplectic form are called \textbf{regular} or \textbf{non-degenerate} when the map
$$TM \rightarrow \bigwedge^k M, \, v \mapsto \iota_v \omega$$ is a vector bundle monomorphism.
\end{Def}

The canonical example of a multisymplectic manifold is the one generalizing the cotangent bundle from mechanics to a multi-cotangent bundle:
\begin{example} Let $Y$ be a manifold and define (for $k \leq \dim Y$) $$M:= \bigwedge^k Y.$$ There is a canonical $k$-form defined on $M$, the \textbf{Liouville} $k$-form, defined as 
$$\Theta|_{\alpha}(v_1, \dots, v_k) := \alpha(\tau_\ast v_1, \dots, \tau_\ast v_k),$$
where $\tau: \bigwedge^k Y \rightarrow Y$ denotes the canonical projection. Then, 
$$\Omega := \d{\Theta}$$ is a closed $(k+1)$-form and thus endows $M$ with a canonical multisymplectic structure. To obtain its local form, notice that for any set of coordinates $(y^i)$ on $Y$, we can define the induced coordinates on $M$, $(y^i, p_{i_1, \dots, i_k}),$ representing the form
$$\alpha = p_{i_1, \dots, i_k} \d{y^{i_1}} \wedge \cdots \wedge \d{y^{i_k}}.$$ Then, the canonical multisymplectic form takes the local expression
$$\Omega = \d{p}_{i_1, \dots, i_k} \wedge \d{y^{i_1}} \wedge \cdots \wedge \d{y^{i_k}}.$$
\end{example}
The main example of multisymplectic manifolds arising in the study of classical field theory is the following
\begin{example}
\label{ex:Classical_field_theory}
The extended Hamiltonian formalism of classical field theory occurs in 
    $$M := \bigwedge^n_2 Y  = \{ \alpha \in \bigwedge^n Y, \iota_{e_1 \wedge e_2} \alpha = 0, e_1, e_2 \in \ker d \pi\},$$
    where $$\pi: Y \rightarrow X$$ is a fibered manifold, and $n = \dim X.$ The manifold $\bigwedge^n_2 Y$ can be endowed with a non-degenerate multisymplectic structure by restricting the canonical multisymplectic form on $\bigwedge^k Y.$ Making abuse of notation, we will still denote by $\Omega$ the restriction. Introducing canonical coordinates $(x^\mu, y^i, p^\mu_i, p)$, where $\pi(x^\mu, y^i) = (x^\mu)$ are fibered coordinates, representing the form 
    $$\alpha = p \d{^n} x + p^\mu_i \d{y^i} \wedge \d{^{n-1}} x_\mu,$$ we have
    $$\Omega = \d{p} \wedge \d{^n} x + \d{p^\mu_i} \wedge \d{y^i} \wedge \d{^{n-1}} x_\mu.$$ A straight-forward calculation shows that this multisymplectic structure is non-degenerate. For an in-depth treatment of classical field theory using multisymplectic geometry we refer to \cite{Binz1988}.
\end{example}

The natural analogue of a Hamiltonian vector field is that of a Hamiltonian multivector field:

\begin{Def}[Hamiltonian multivector field, Hamiltonian form] Let $(M, \omega)$ be a multisymplectic manifold. Then, a multivector field $U \in \mathfrak{X}^p(M)$ is called \textbf{Hamiltonian} if $$\iota_U \omega = \d{\alpha},$$ for certain $(k-p)-$form $\alpha \in \Omega^{k-p}(M)$ which is called \textbf{Hamiltonian} as-well.
\end{Def}

For details on the study of multisymplectic manifolds, we refer to \cite{Ryvkin2019, Gotay2004, Cantrijn1996, Cantrijn1999}.
\subsection{Dirac structures}
\label{Graded_Dirac_Strcutures:Subsection:Dirac_Structures}

Dirac structures, first introduced independently by Courant \cite{Courant1990} and Dorfman \cite{dorfman}, are a generalization of pre-symplectic manifolds (manifolds equipped with
a closed $2$-form $\omega$, see \cite{Abraham2008,Arnold1989}) and Poisson manifolds (manifolds equipped with an integrable bivector field $\Lambda$, 
see \cite{Vaisman1994}). Dirac geometry has applications to the theory of constraints
in classical mechanics. Indeed, in general, a submanifold of a Poisson manifold does not inherit a Poisson structure; however, under certain 
mild assumptions, it inherits a Dirac structure. For a swift introduction to Dirac geometry, see \cite{Bursztyn2011}.\\

A Dirac structure on manifold $M$ is a vector subbundle of $ E:= TM \oplus_M T^\ast M$. There are two important operations related to this vector 
bundle. The first one, a natural point-wise pairing: given $x \in M$, and  $(u, \alpha), (v, \beta) \in E_x = T_x M \oplus T_x^\ast M,$ 
we can define the symmetric bracket

\begin{equation} \label{Dirac_structures:Pointwise_pairing}
    \llangle (u, \alpha), (v, \beta) \rrangle := \beta(u) + \alpha(v).
\end{equation}
The second one, an operation on sections, the \textbf{Courant bracket} 
$$\Gamma(E) \otimes \Gamma(E) \xrightarrow{\llbracket\cdot, \cdot\rrbracket} \Gamma(E) $$
defined as  

\begin{equation} \label{Dirac_structures:Courant_bracket}
    \llbracket (u, \alpha), (v, \beta)\rrbracket := ([u, v], \pounds_u \beta - \pounds_v \alpha + \frac{1}{2} d (\alpha(v) - \beta(u))).
\end{equation}

\begin{Def}[Dirac structure] A \textbf{Dirac structure} on a manifold $M$ is a vector subbundle 
\[\begin{tikzcd}
	L & {TM \oplus_MT^\ast M} \\
	M
	\arrow[hook, from=1-1, to=1-2]
	\arrow[from=1-1, to=2-1]
	\arrow[curve={height=-6pt}, from=1-2, to=2-1]
\end{tikzcd}\]
    satisfying the following properties:
    \begin{enumerate}[i)]
        \item It is \textit{Lagrangian} with respect to the bracket of \cref{Dirac_structures:Pointwise_pairing}, that is, 
        $$L = L^\perp := \{(v, \beta) \in E: \,\, \llangle (v, \beta), (u, \alpha) \rrangle  = 0, \forall (u, \alpha) \in L\}.$$
        \item It is  \textit{involutive} (i.e. closed) with respect to the Courant bracket (\cref{Dirac_structures:Courant_bracket}).
    \end{enumerate}
\end{Def}

There has been several attempts to generalize the concepts of \cref{Graded_Dirac_Strcutures:Subsection:Dirac_Structures}
to the realm of field theories (multisymplectic geometry), which we briefly review now.

\subsection{Multi-Dirac structures}
\label{Overview:subsection:multidirac}

The first definition was proposed by J. Vankershaver, H. Yoshimura and M. Leok in \cite{Vankerschaver2011, Vankerschaver2012}. 
Motivated by the graded nature of multisymplectic geometry, instead of using the bundle $E = TM \oplus_M T^\ast M,$ 
they defined the family $$E_p := \bigvee_p M \oplus_M \bigwedge^{k+1-p} M,$$ where $1 \leq p \leq k,$ and $k$ is a \textit{fixed} integer,
which will be the degree of the multi-Dirac structure. Note that the particular case $k = 1$ recovers the bundle of Dirac geometry, $E$.\\

We also have two brackets, which are the generalization of the point-wise pairing and the Courant bracket. Let $k\geq 1$ be a fixed integer, 
and $1 \leq p , q \leq k$ such that $p +q\leq k+1$.\\
\begin{enumerate}
    \item The \textit{graded point-wise pairing} $$E_p \otimes_M E_q \xrightarrow{\llangle \cdot, \cdot \rrangle} \bigwedge^{k+1-(p+q)} M$$
    at $x \in M$ for $(U,\alpha) \in (E_p)\big |_x, (V, \beta) \in (E_q)\big |_x$ is given by 
    \begin{equation}\label{Graded_Dirac_structures:Graded_Point-wise_pairing}
        \llangle (U, \alpha), (V, \beta) \rrangle = \iota_U \beta  - (-1)^{pq} \iota_V \alpha.
    \end{equation}
    \item The \textit{graded Courant bracket} $$\Gamma(E_p) \otimes \Gamma(E_q) \xrightarrow{\llbracket\cdot, \cdot\rrbracket} \Gamma(E_{p+q - 1})$$ 
    is given by 
    \begin{equation}\label{Graded_Dirac_structures:Graded_Courant_bracket}
        \llbracket (U, \alpha), (V, \beta) \rrbracket := \left([U, V], (-1)^{(p-1)q} \pounds_U \beta +
         (-1)^{q} \pounds_V \alpha - (-1)^q \frac{1}{2} d( \iota_V \alpha  + (-1)^{pq} \iota_U \beta)\right).
    \end{equation}
\end{enumerate}

\begin{Def}[Multi-Dirac structure] A \textbf{multi-Dirac structure} of degree $k$ on a manifold $M$ is a sequence of vector subbundles 
    $D_1, \dots, D_k$ 
    \[\begin{tikzcd}
        {D_p} & {E_p = \bigvee_p M \oplus_M \bigwedge^{k+1-p}M} \\
        M
        \arrow[hook, from=1-1, to=1-2]
        \arrow[from=1-1, to=2-1]
        \arrow[curve={height=-12pt}, from=1-2, to=2-1]
    \end{tikzcd}\]
    satisfying:
    \begin{enumerate}[i)]
        \item It is \textit{Lagrangian} with respect to the bracket defined in \cref{Graded_Dirac_structures:Graded_Point-wise_pairing}, that is,
        $$D^q = (D^p)^{\perp, q} := \{(V, \beta) \in E^q: \,\, \llangle (V, \beta), (U, \alpha) \rrangle = 0, \forall (U, \alpha) \in D^p\},$$
        for all $p, q$ such that $p + q \leq k+1.$ 
        \item It is \textit{involutive} with respect to the Courant bracket (\cref{Graded_Dirac_structures:Graded_Courant_bracket}), that is, 
        $\llbracket \cdot, \cdot \rrbracket$ restricts to an operation on sections 
        $$\Gamma(D_p) \otimes \Gamma(D_q) \xrightarrow{\llbracket\cdot, \cdot\rrbracket} \Gamma(D_{p+q-1}).$$
    \end{enumerate}
\end{Def}

Multi-Dirac structures generalize multisymplectic structures. Indeed, we have the following result.

\begin{proposition}[\cite{Vankerschaver2011}] Let $\omega \in \Omega^{k+1}(M)$ be a $(k+1)$-form. 
    Then, the sequence of vector subbundles $$D_p := \{(U, \iota_U \omega): \,\, U \in \bigvee_p M\}$$ defines a multi-Dirac structure
    of order $k$ if and only
    if $\d{\omega} = 0.$
\end{proposition}
\begin{proof}
    First, we will prove that it is Lagrangian. Indeed, given two multivector fields $U \in \mathfrak{X}^p(M), V \in \mathfrak{X}^q(M),$ we have
    \begin{align*}
        \llangle (U, \iota_U \omega), (V, \iota_V \omega)\rrangle &= \iota_U \iota_V \omega - (-1)^{pq} \iota_V \iota_U \omega = 0,
    \end{align*}
    which yields $$ D_q \subseteq (D_p)^{\perp, q}.$$ Now, conversely, if $(V, \beta) \in E_q$ satisfies
    $$\llangle (U, \iota_U \omega), (V, \iota_V \omega)\rrangle = 0,$$
    for any multivector field $U \in \mathfrak{X}^p(M)$, then,
    \begin{align*}
        0 = \iota_U \beta - (-1)^{pq} \iota_V \iota_U \omega =  \iota_U \beta -  \iota_U \iota_V \omega,
    \end{align*}
    and, since $U$ is arbitrary, this implies $\beta = \iota_V \omega$, proving that
    $$
    D_q = (D_p) ^{\perp,q}.
    $$
    Now, it only remains to show that it is involutive if and only if $\d{\omega} = 0.$ Recall that
    $$\iota_{[U,V]} \omega = (-1)^{(p-1)q} \pounds_U \iota_V \omega - \iota_V \pounds_U \omega.$$
    Now, $D_p$ is integrable if and only if
    \begin{align*}
        \iota_{[U, V]} \omega &= (-1)^{(p-1)q} \pounds_U \iota_V \omega +
         (-1)^{q} \pounds_V \iota_U \omega - (-1)^q \frac{1}{2} d( \iota_V \iota_U \omega  + (-1)^{pq} \iota_U \iota_V \omega)\\
         &= (-1)^{(p-1)q} \pounds_U \iota_V \omega +
         (-1)^{q} \pounds_V \iota_U \omega - (-1)^q  d \iota_V \iota_U \omega\\
         &= (-1)^{(p-1)q} \pounds_U \iota_V \omega - \iota_V \d{\iota_U} \omega,
    \end{align*}
    that is, if and only if 
    \begin{align*}
        (-1)^{(p-1)q} \pounds_U \iota_V \omega - \iota_V \pounds_U \omega 
        = (-1)^{(p-1)q} \pounds_U \iota_V \omega - \iota_V \d{\iota_U} \omega.
    \end{align*}
    It is clear that the previous equation holds for every $U \in \mathfrak{X}^p(M)$, and $V \in \mathfrak{X}^q(M)$
    only when
    $$\iota_V \iota_U \d{\omega} = 0,$$
    which is equivalent to $\omega$ being closed.
\end{proof}

\subsection{Higher Dirac structures}
\label{Overview:subsection:higher_Dirac}

In \cite{Zambon2012}, M. Zambon noticed that given a multi-Dirac structure of order $k$ on a manifold $M$, $D_1, \dots, D_k$, $D_1$
determines completely the multi-Dirac structure by the equality $D^p = (D^1)^{\perp, p},$ for all $1 \leq p \leq k.$ He then introduced
the following concept.

\begin{Def}[Higher Dirac structure] A \textbf{higher Dirac structure} of order $k$ on a manifold $M$ is a vector subbundle
\[\begin{tikzcd}
	D & {E_1 = TM \oplus_M \bigwedge^{k}M} \\
	M
	\arrow[hook, from=1-1, to=1-2]
	\arrow[from=1-1, to=2-1]
	\arrow[curve={height=-12pt}, from=1-2, to=2-1]
\end{tikzcd}\]
satisfying
\begin{enumerate}[i)]
    \item It is \textit{Lagrangian}, that is $D = D ^{\perp, 1}.$
    \item It is \textit{involutive} with respect to the Courant bracket.
\end{enumerate}
\end{Def}

M. Zambon also proved the equivalence (under some mild assumptions) of multi-Dirac structures and higher Dirac structures.
\begin{theorem}[\cite{Zambon2012}]
    \label{Dirac_Structures:Overview:Thm:Equivalence_HigherDirac_multiDirac}
    Let $M$ be a manifold and fix $k \geq 1$. Then, there is an injective mapping
    $$\{\text{Multi-Dirac structures of order $k$ on M}\} \hookrightarrow \{\text{Higher Dirac structures of order $k$ on $M$}\}$$
    given by $$D_1, \dots, D_k \mapsto D  = D_1$$ which is a bijection onto the set of higher Dirac structures $D \subseteq E^1$ such that
    $D^{\perp, p}$ defines a vector subbundle for $1 \leq p \leq k$.
\end{theorem}


\subsection{Weak higher Dirac structures} \label{Overview:subsection:Weak_Higher_Dirac}

We know that multi-Dirac and higher Dirac structures can be thought of as a generalization of multisymplectic geometry. Whether
these kind of geometries can be thought of as a generalization of some sort of ``higher Poisson structure'' was studied in \cite{Bursztyn2019}.\\

Earlier, in \cite{Bursztyn2011}, H. Bursztyn proposed the following definition for a generalization of Poisson structure to the realm of
field theories.

\begin{Def}[Higher Poisson structure]
    \label{Dirac_structures:Overview:Definition:Weak-Higher-Dirac-Structure}
    A \textbf{higher Poisson structure} of order $k$ on a manifold $M$ is a pair $(S,\sharp)$, where
    $S$ is a vector subbundle 
    \[\begin{tikzcd}
        S & {\bigwedge^kM} \\
        M
        \arrow[hook, from=1-1, to=1-2]
        \arrow[from=1-1, to=2-1]
        \arrow[curve={height=-12pt}, from=1-2, to=2-1]
    \end{tikzcd},\]
    and $\sharp$ is a vector bundle mapping $$\sharp: S \rightarrow TM$$ satisfying:
    \begin{enumerate}[i)]
        \item $S^{\circ, 1} = 0$, where $S^{\circ, 1}$ denotes the first annihilator. In general, for a family of forms $S \subseteq \bigwedge^a M$,
        and $p \leq a$, we define $$S^{\circ, p} := \{U \in \bigvee_p M:\,\, \iota_U \alpha = 0, \forall \alpha \in S\}.$$
        \item $\sharp$ is \textit{skew-symmetric}, that is, $\iota_{\sharp(\alpha)} \beta = - \iota_{\sharp(\beta)} \alpha,$ for all $\alpha, \beta \in S.$
        \item $S$ is \textit{involutive} with respect to the bracket
        $$[\alpha, \beta] := \pounds_{\sharp(\alpha)} \beta - \iota_{\sharp(\alpha)} \d{\beta}$$
        and satisfies $$\sharp([\alpha, \beta]) = [\sharp(\alpha), \sharp(\beta)].$$
    \end{enumerate} 
\end{Def}

Notice that, for $k = 1$, we recover the notion of a classical Poisson structure, where $\sharp$ is the morphism induced by the bivector $\Lambda$.
Indeed, for $k = 1$, the condition $S^{\circ, 1} = 0$ is equivalent to $S = T^\ast M$ which, together with skew-symmetry of $\sharp$, implies
that it is the contraction with a bivector field $\Lambda$. Involutivity is characterized by 
$$[\Lambda, \Lambda] = 0,$$ which is exactly the notion of a Poisson manifold.\\

As we mentioned above, Poisson manifolds are a particular case of Dirac manifolds defining the vector subbundle as its graph
$$D := \left\langle (\sharp(\alpha), \alpha) \in E_1, \, \alpha \in T^\ast M \right\rangle.$$
But, in \cite{Bursztyn2019}, H. Bursztyn, N. Martinez-Alba and R. Rubio observed that, in general, the graph of a higher Poisson structure 
$$D := \{(\sharp(\alpha), \alpha): \,\, \alpha \in S\} \subseteq E_1 = TM \oplus_M \bigwedge^k M$$
fails to define a higher Dirac structure. Nevertheless, they noticed that it does define a \textit{weak higher Dirac structure} (what they
simply call higher Dirac).

\begin{Def}[Weak higher Dirac structure] A weak higher Dirac structure of degree $k$ is a vector subbundle $D \subseteq E_1$ satisfying
    \begin{enumerate}[i)]
        \item It is \textit{weakly Lagrangian,} that is, it is isotropic $$D \subseteq D^{\perp,1}$$ and $$D \cap TM = (\operatorname{pr}_2(D))^{\circ, 1},$$ where 
        $$\operatorname{pr_2}: E_1 \rightarrow \bigwedge^k M$$ is the projection onto the second factor.
        \item It is \textit{involutive} with respect to the Courant bracket.
    \end{enumerate}
\end{Def}

In the case $k = 1$, this notion and the one introduced by M. Zambon coincide, but for $k \geq 2$ it is slightly more general.\\


The weak analogue to multi-Dirac structures (in the sense of \cite{Vankerschaver2011,Vankerschaver2012})
is yet to be defined and studied. The relevance of finding the so-called \textit{graded Dirac} structure is that multisymplectic manifolds $(M, \omega)$
are naturally equipped with a graded Lie algebra, having as graded vector space Hamiltonian forms of arbitrary degree (see \cite{Cantrijn1999,Cantrijn1996}). 
We begin by studying the linear case in the next section.

\section{Graded Dirac structures}
\label{section:Graded_Dirac_structures}


\subsection{Linear graded Dirac structures}
\label{Graded_Dirac_Strcutures:Subsection:Linear_Graded_Dirac_Structures}

Let $V$ be a \textit{finite dimensional} vector space and $k\geq 1$ be a fixed integer. Following the notation of \cref{section:Overview},
define 
$$
E_p := \bigwedge^p V \oplus \bigwedge^{k+1-p} V^\ast,
$$ 
for $1 \leq p \leq k$. Next, take $p, q$ such that $p + q \leq k+1$, then we have the graded-symmetric pairing 
$$E_p \otimes E_q \xrightarrow{\llangle \cdot, \cdot \rrangle} \bigwedge^{k+1-(p+q)} V^\ast$$ defined by \cref{Graded_Dirac_structures:Graded_Point-wise_pairing}.\\

The natural weak analogue of multi-Dirac structures at the linear level is the following
\begin{Def}[Graded Dirac structure] A \textbf{linear graded Dirac structure} of order $k$ on a vector space $V$ is a sequence of subspaces $D_1, \dots, D_k$, $D_p \subseteq E_p$
    satisfying the following property:
    \begin{enumerate}[i)]
        \item The sequence is \textit{weakly Lagrangian}, that is, it is isotropic, 
        $$D^p \subseteq (D^q)^{\perp, p}$$ for all $p,q$ such that $p +q \leq k+1;$ and it satisfies
        $$D^q \cap \bigwedge^q V = \left(\operatorname{pr}_2(D^p) \right)^{\circ, q},$$
        for all $p, q$ such that $p +q \leq k+1$.
    \end{enumerate}
\end{Def}

Let us also define what we mean by a linear weak higher Dirac structure.

\begin{Def}A \textbf{linear weak higher Dirac structure} of order $k$ on a vector space $V$ is a subspace $D \subseteq E_1$
    satisfying:
    \begin{enumerate}[i)]
        \item It is \textit{weakly Lagrangian}, that is, it is isotropic, 
        $$D \subseteq (D)^{\perp, 1};$$ and it satisfies
        $$D \cap V = \left(\operatorname{pr}_2(D) \right)^{\circ, 1}.$$
    \end{enumerate}
\end{Def}

\begin{remark} Clearly, if $D_1, \dots, D_k$ is a graded Dirac structure of order $k$ on $V$, $D_1$ is a linear weak higher Dirac structure
    of order $k$ on $V$.
\end{remark}

Hence, we get a well defined mapping 
$$
\{\text{Linear graded Dirac structures}\} \rightarrow \{\text{Linear weak higher Dirac structures}\}$$
given by 
$$
D_1, \dots, D_k \mapsto D_1.
$$
The rest of this subsection is concerned with proving that this mapping is a bijection, 
giving the linear weak analogue of \cref{Dirac_Structures:Overview:Thm:Equivalence_HigherDirac_multiDirac}. The main idea is giving a different description
of graded Dirac structures. The proof is rather technical, and it is not necessary to understand the rest of the text, although it does uses ideas from the description of graded Poisson structures, as the reader may see (\cref{section:Graded_Poisson_structures}).

\begin{theorem}[Equivalence of linear graded Dirac structures and linear weak higher Dirac structures] 
    \label{Linear_Dirac_Structures:thm:Linear_equivalence}
    The mapping $$\{\text{Linear graded Dirac structures of order }k\} \rightarrow \{\text{Linear weak higher Dirac structures of order }k\}$$
    given by $$D_1, \dots, D_k \mapsto D_1$$
    defines a one-to-one correspondence.
\end{theorem}

\begin{proposition}[An equivalent description of linear graded Dirac structures] \label{Dirac_Structures:Linear_Dirac_Structures:prop:EquivalentDescription_GradedDirac}
    A linear graded Dirac structure 
    of order $k$, $D_1, \dots, D_k,$ is equivalent to a family of vector subspaces and mappings
    \begin{enumerate}[i)]
        \item $S^{k+1-p} \subseteq \bigwedge^{k+1-p} V^\ast,$ $p = 1, \dots, k$;
        \item $K_p \subseteq \bigwedge^{p} V$,
        \item $\sharp_{k+1-p}: S^{k+1-p} \rightarrow \left(\bigwedge^{p} V \right)/ K_p,$
    \end{enumerate}
    for $1 \leq p \leq k+1$ satisfying the following properties:
    \begin{enumerate}[i)]
        \item \label{Linear_Dirac_Structures:prop:equivalentDiracStrcutures:enumerate:1}
        $K_p = (S^{k+1-q})^{\circ,p},$ for all $p,q$ such that $p + q \leq k+1;$
        \item \label{Linear_Dirac_Structures:prop:equivalentDiracStrcutures:enumerate:2}
        $\iota_{\sharp_{k+1-p}(\alpha)} \beta = (-1)^{pq}\iota_{\sharp_{k+1-q}(\beta)}\alpha,$
        for $\alpha \in S^{k+1-p},$ $\beta \in S^{k+1-q}.$
    \end{enumerate}
    This correspondence is given by the `graph' of $\sharp_a$ 
    $$D_p := \{ (U , \alpha): \alpha \in S^{k+1-p}, \, \sharp_{k+1-p}(\alpha) = U + K_p\}.$$
\end{proposition}

\begin{proof} Let us first prove that a family of vector subspaces and mappings $(S^{k+1-p}, K_p, \sharp_{k+1- p})$ 
    satisfying the previous conditions for $1 \leq p \leq k+1$ defines a linear graded Dirac structure. We only need to check that
    the sequence defined by $$D_p := \{ (U , \alpha): \alpha \in S^{k+1-p}, \, \sharp_{k+1-p}(\alpha) = U + K_p\}$$ is weakly Lagrangian.
    It is clearly isotropic. Indeed, let $(U, \alpha) \in D_p$, $(V, \beta) \in D_q$ for $p,q$ such that $p+q \leq k+1.$ Then, we have that
    $$\sharp_{k+1-p}(\alpha) = U + K_p, \,\, \sharp_{k+1-q}(\beta) = V + K_q$$ and thus
    $$\llangle (U, \alpha), (V, \beta)\rrangle = \iota_{\sharp_{k+1-p}(\alpha)} \beta - (-1)^{pq}\iota_{\sharp_{k+1-p}(\beta)} \alpha = 0,$$
    by \ref{Linear_Dirac_Structures:prop:equivalentDiracStrcutures:enumerate:2}.
    Finally, since $$D_p \cap \bigwedge^p V = K_p , \,\, \operatorname{pr}_2(D_q) = S^{k+1-q},$$
    using \ref{Linear_Dirac_Structures:prop:equivalentDiracStrcutures:enumerate:1} we get $$D_q \cap \bigwedge^q V = \left(\operatorname{pr}_2(D_p) \right)^{\circ, q},$$
    for all $p,q$ such that $p+q \leq k+1,$ which proves that it is weakly Lagrangian. \\

    Conversely, let $D_1, \dots, D_k$ be a linear graded Dirac structure on $V$. Define 
    $$K_p := D_p \cap \bigwedge^p V,$$ and 
    $$S^{k+1-p} := \{\alpha \in \bigwedge^{k+1-p} V^\ast, \,\, \exists u \in K_p, \, u + \alpha \in D_p\}.$$
    If we put $$\sharp_{k+1-p}(\alpha) := u + K_p,$$ for certain $u$ such that $u + \alpha \in D_p,$ then we get
    the desired structure, as one can easily check.
\end{proof}

The same line of reasoning yields the following result

\begin{proposition}\label{Dirac_Structures:Linear_Dirac_Structures:prop:EquivalentDescriptionHigherDirac}
    A linear higher Dirac structure $D \subseteq V \oplus \bigwedge^k V^\ast$ is equivalent
    to a choice of subspaces $S \subseteq \bigwedge^k V^\ast$, $K\subseteq V$ and a mapping
    $$\sharp: S \rightarrow V/K$$ satisfying
    \begin{enumerate}[a)]
        \item $K = S^{\circ, 1};$
        \item $\iota_{\sharp(\alpha)} \beta = - \iota_{\sharp(\beta)} \alpha,$ for all $\alpha, \beta \in S.$
    \end{enumerate}
    The correspondence is given by $$D := \{(U, \alpha) \in E_1 : \,\, \sharp(\alpha) = U + K\}.$$
\end{proposition}
\begin{proof} It follows from a similar procedure as in \cref{Dirac_Structures:Linear_Dirac_Structures:prop:EquivalentDescription_GradedDirac}.
\end{proof}

Now we are ready to prove \cref{Linear_Dirac_Structures:thm:Linear_equivalence}.

\begin{proof}(of \cref{Linear_Dirac_Structures:thm:Linear_equivalence}) By 
    \cref{Dirac_Structures:Linear_Dirac_Structures:prop:EquivalentDescription_GradedDirac} and 
    \cref{Dirac_Structures:Linear_Dirac_Structures:prop:EquivalentDescriptionHigherDirac}, to prove that the mapping is both an
    injection and a surjection, we need to check that given two
    subspaces $S \subseteq \bigwedge^k V^\ast$, $K \subseteq V$, and a mapping $$\sharp: S \rightarrow V/K$$ satisfying 
    \begin{enumerate}[a)]
        \item $K = S^{\circ, 1};$
        \item $\iota_{\sharp(\alpha)} \beta = - \iota_{\sharp(\beta)} \alpha,$ for all $\alpha, \beta \in S;$
    \end{enumerate}
    there exists an \textit{unique} family of vector subspaces $S^a\subseteq \bigwedge^aV^\ast$, $K_p \subseteq \bigwedge^p V$ and mappings 
    $$\sharp_{a}: S^a \rightarrow \bigwedge^{k+1-a} V / K_{k+1-a}$$ satisfying
    \begin{enumerate}[(i)]
        \item \label{proof_condition1}
        $K_p = (S^{k+1-q})^{\circ,p},$ for all $p,q$ such that $p + q \leq k+1;$
        \item \label{proof_condition2}
        $\iota_{\sharp_{k+1-p}(\alpha)} \beta = (-1)^{pq}\iota_{\sharp_{k+1-q}(\beta)}\alpha,$
        for $\alpha \in S^{k+1-p},$ $\beta \in S^{k+1-q};$
    \end{enumerate}
    and such that $$S = S^k, \, K = K_1, \,\sharp = \sharp_k.$$
    Let $(S, K, \sharp)$ be given as above. We divide the proof in five steps.
    \begin{center}
        \textit{\underline{1. Definition of $K_p$}}
    \end{center}
    It is clear that we need to define 
    $$K_p := (S)^{\circ, p},$$
    by condition \ref{proof_condition1}. 
    \begin{center}
        \textit{\underline{2. Definition of $S^a$}}
    \end{center}
    Now, to define $S^a \subseteq \bigwedge^a V^\ast,$ notice that condition \ref{proof_condition2}
    implies, for $p = a$, $q = k+1-a$, $$K_a = (S^a)^{\circ,a}.$$ Consequently, since $(K_a)^{\circ, a}$ can be thought of
    as the annihilator of $K_a$ in $(\bigwedge^{a} V)^\ast = \bigwedge^{a} V^\ast,$ we have
    $$(K_a)^{\circ,a} = ((S) ^{\circ, a})^{\circ,a} = S^a,$$
    which determines $S^a$ for each $1 \leq a \leq k+1$.  Now it only remains to define the mappings
    $$\sharp_a : S^a \rightarrow \bigwedge^{k+1-a} V/ K_{k+1-a}.$$
    \begin{center}
        \textit{\underline{3. Description of the subspaces $S^a$}}
    \end{center}
    First, notice that the previous description of the
    subspaces $S^a$ implies that we have the inclusion
    $$\left \langle \iota_U \alpha: U \in\bigwedge^{k-a} V^\ast, \,\alpha \in S \right \rangle \subseteq S^a.$$ Now, let us check that 
    we have the equality.

    \begin{lemma}[Auxiliary Lemma] \label{Linear_Graded_Dirac_structures:Auxiliary_Lemma_II} Let $V$ be a \textit{finite dimensional} vector space,
        $k \leq \dim V$, and $S \subseteq \bigwedge^k V^\ast$ be a subspace of forms. Then, for any $1 \leq a \leq k$ we have
        $$\left \langle \iota_U \alpha: U \in\bigwedge^{k-a} V, \,\alpha \in S \right \rangle = (K_a)^{\circ, a},$$
        where $K_a = S^{\circ, a}.$
    \end{lemma}
    \begin{proof}
        Let us first study the case where $S$ is generated by an unique $k$-form, $\beta$: $$S = \langle \beta \rangle.$$
        Notice that $(K_a)^{\circ, a}$ can be identified as the annihilator of $K_a$ in $\bigwedge^a V$
        (indeed, $\left(\bigwedge^a V\right) ^\ast = \bigwedge^a V^\ast$). Define
        $$\Phi_\beta: \bigwedge^a V \rightarrow \bigwedge^{k-a} V^\ast$$ by 
        $$\Phi_\beta(U) := \iota_U \beta.$$ Then, the dual mapping
        $$\Phi_\beta^\ast: \bigwedge^{k-a} V \rightarrow \bigwedge^a V$$ is given by
        $$\Phi_\beta^\ast(U) = (-1)^{a(k-a)} \iota_U \beta,$$ as an easy computation shows. Then,
        using the following equality from linear algebra\footnote{In general, for a linear 
        mapping between two finite dimensional vector spaces $f: V \rightarrow W$, we have $\Im f^\ast = (\ker f)^{\circ}$.},
        $$(\ker \Phi_\beta)^{\circ, a} = \Im \Phi_\beta^\ast,$$ we get
        $$K_a^{\circ, a} = \langle \iota_U \beta: U\in \bigwedge^{k-a} V \rangle,$$
        which is exactly what we wanted to prove.\\

        Now, for the general case, suppose $$S = \langle \beta_1, \dots, \beta_l \rangle.$$ 
        Then, $$K_a = S^{\circ, a} = \langle\beta_1\rangle^{\circ, a} \cap \cdots \cap \langle\beta_l\rangle^{\circ, a}$$
        and we have
        \begin{align*}
            (K_a)^{\circ, a}& = \left( \langle\beta_1\rangle^{\circ, a} \cap \cdots \cap \langle\beta_l\rangle^{\circ, a}\right)^{\circ, a}\\
            &= ( \langle\beta_1\rangle^{\circ, a})^{\circ, a} + \cdots + ( \langle\beta_l\rangle^{\circ, a})^{\circ, a}\\
            &= \left \langle \iota_U \beta_1: U \in\bigwedge^{k-a} V\right \rangle + \cdots 
            + \left \langle \iota_U \beta_l: U \in\bigwedge^{k-a} V\right \rangle\\
            &= \left \langle \iota_U \alpha: U \in\bigwedge^{k-a} V, \,\alpha \in S \right \rangle,
        \end{align*}
        proving the Lemma.
    \end{proof}

    Using \cref{Linear_Graded_Dirac_structures:Auxiliary_Lemma_II}, we get the following description of the subspaces $S^a$,
    $$S^a = \left \langle \iota_U \alpha: U \in\bigwedge^{k-a} V^\ast, \,\alpha \in S \right \rangle.$$

    \begin{center}
        \textit{\underline{4. Definition of the maps $\sharp_a$}}
    \end{center}

    With this description, we can define 
    the mappings $$\sharp_{a}: S^a \rightarrow \bigwedge^{k+1-a} V/ K_{k+1-a}.$$
    Let $\iota_U \alpha \in S^a$, for certain $U \in \bigwedge^{k-a} V,$ $\alpha \in S$. Then, for every $\beta \in S$, 
    $\sharp_a(\iota_U \alpha)$ must satisfy :
    \begin{align*}
        \iota_{\sharp_a(\iota_U \alpha)} \beta &= (-1)^{k+1-a} \iota_{\sharp(\beta)} \iota_U \alpha =
        (-1)^{k+1-a}\cdot(-1)^{k-a} \iota_U \iota_{\sharp(\beta)} \alpha \\
        &= - \iota_U \iota_{\sharp(\beta)} \alpha =  \iota_U \iota_{\sharp(\alpha)} \beta = \iota_{\sharp(\alpha) \wedge U} \beta.
    \end{align*}
    Since this equality holds for every $\beta \in S,$ we are forced to define 
    $$\sharp_a(\iota_U \alpha) := \sharp(\alpha) \wedge U + K_{k+1-a}.$$ Let us check that this mapping is well-defined.
    Indeed, suppose $\iota_U \alpha = \iota_V \gamma,$ for certain $U, V \in \bigwedge^{k-a} V, \alpha, \gamma \in S.$
    Then, arbitrary $\beta \in S,$ and using the properties of $\sharp$ we have
    \begin{align*}
        \iota_{(\sharp(\alpha) \wedge U - \sharp(\gamma) \wedge V)} \beta = (-1)^{k+1-a} \iota_{\sharp(\beta)} \left(\iota_U \alpha - \iota_V \gamma \right) = 0.
    \end{align*}
    Since this holds for all $\beta \in S$, we conclude that $\sharp(\alpha) \wedge U - \sharp(\gamma) \wedge V \in K_{k+1-p}$, proving our assertion.\\
    
    \begin{center}
        \textit{\underline{5. $(S^a, K_p, \sharp_a)$ defines a linear graded Dirac structure}}
    \end{center}
    It only remains to show that the triple $(S^a, K_p, \sharp_a)$ defines a linear graded Dirac structure, that is, it satisfies
    the conditions of \cref{Dirac_Structures:Linear_Dirac_Structures:prop:EquivalentDescription_GradedDirac}. Property 
    \ref{Linear_Dirac_Structures:prop:equivalentDiracStrcutures:enumerate:1} is clear. Indeed, for $p \leq a \leq k$ and using the
    description of $S^a$ given in \textit{\underline{Step 3}}, 
    \begin{align*}
        (S^a)^{\circ, p} &= \langle V \in \bigwedge^{p} V: \iota_V \iota_U \alpha = 0, \forall U \in \bigwedge^{k-a}V, \alpha \in S \rangle\\
        &= \langle V \in \bigwedge^{p} V: \iota_V\alpha = 0, \forall  \alpha \in S \rangle = K_p.
    \end{align*}
    Finally, for property \ref{Linear_Dirac_Structures:prop:equivalentDiracStrcutures:enumerate:2}, given 
    $U \in \bigwedge^{k-a} V, W \in \bigwedge^{k-b} V,$ $\alpha, \beta \in S:$
    \begin{align*}
        \iota_{\sharp_{a}(\iota_U \alpha)} (\iota_W \beta) &= \iota_{\sharp(\alpha) \wedge U} (\iota_W \beta) = (-1)^{(k+1-a)(k-b)} \iota_W \iota_U \iota_{\sharp(\alpha)} \beta\\
        &= (-1)^{(k+1-a)(k-b) +1}\iota_W \iota_U \iota_{\sharp(\beta)} \alpha =(-1)^{(k+1-a)(k-b) +1 k k - a}\iota_W \iota_{\sharp(\beta)} \iota_U \alpha \\
        &= (-1)^{(k+1-a)(k+1-b)} \iota_{\sharp_b(\iota_W \beta)}(\iota_U \alpha),
    \end{align*}
    finishing the proof.
\end{proof}

We have the following immediate corollary:

\begin{corollary}\label{Linear_Graded_Dirac:corollary:Dp_in_terms_of_D1}
    Let $D_1, \dots, D_k$ be a linear graded Dirac structure of order $k$ on $V$. Then, we have that
    $$D_p = \langle (W\wedge U, \iota_U \alpha), U \in \bigwedge^{p-1} V, \,(W, \alpha)\in D_1\rangle.$$
\end{corollary}


\subsection{Graded Dirac manifolds}
\label{Graded_Dirac_Strcutures:Subsection:Graded_Dirac_manifolds}
In this subsection we study graded Dirac manifolds, the geometric version of the linear counterpart studied
in \cref{Graded_Dirac_Strcutures:Subsection:Linear_Graded_Dirac_Structures}. Let us begin by not taking into account the
integrability issues:

\begin{Def}[Almost graded Dirac structure] An \textbf{almost graded Dirac structure} of order $k$ on a manifold $M$ 
    is a family of vector subbundles 
    \[\begin{tikzcd}
        {D_p} & {E_p := \bigvee_p M \oplus_M \bigwedge^{k+1-p}M} \\
        M
        \arrow[hook, from=1-1, to=1-2]
        \arrow[from=1-1, to=2-1]
        \arrow[curve={height=-12pt}, from=1-2, to=2-1]
    \end{tikzcd}\]
    that defines a linear graded Dirac structure (see \cref{Graded_Dirac_Strcutures:Subsection:Linear_Graded_Dirac_Structures}) point-wise,
    that is, such that for every $x \in M$, $D_1\big |_x, \dots, D_k \big |_x$ defines a linear graded Dirac structure on $T_xM.$
\end{Def}

For the notion of integrability, we will use the graded Courant bracket
$$\Gamma(E_p) \otimes \Gamma(E_q) \xrightarrow {\llbracket \cdot, \cdot \rrbracket} \Gamma(E_{p+q-1})$$
given by 
$$\llbracket (U, \alpha), (V, \beta) \rrbracket := \left([U, V], (-1)^{(p-1)q} \pounds_U \beta +
         (-1)^{q} \pounds_V \alpha - (-1)^q \frac{1}{2} d( \iota_V\alpha  + (-1)^{pq} \iota_U \beta)\right).$$

\begin{Def}[Graded Dirac structure] A \textbf{graded Dirac structure} of order $k$ on $M$ is an almost graded Dirac structure
    of order $k$ on $M$, $D_1, \dots, D_k$, which is involutive with respect to the Courant bracket, that is,
    such that $\llangle \cdot, \cdot \rrangle$ defines an operation on sections
    $$\Gamma(D_p) \otimes \Gamma(D_q) \xrightarrow{\llbracket\cdot, \cdot \rrbracket} \Gamma(D_{p+q-1}).$$
\end{Def}

Using \cref{Linear_Dirac_Structures:thm:Linear_equivalence}, we get the corresponding equivalence between graded Dirac structures
and higher Dirac structures 
(which generalizes \cref{Dirac_Structures:Overview:Thm:Equivalence_HigherDirac_multiDirac}). First,
let us define what we mean by an almost weak higher Dirac structure:

\begin{Def} An \textbf{almost weak higher Dirac structure} of order $k$ on $M$ is a vector subbundle 
    \[\begin{tikzcd}
        D & {E^1 = TM \oplus_M \bigwedge^kM} \\
        M
        \arrow[hook, from=1-1, to=1-2]
        \arrow[from=1-1, to=2-1]
        \arrow[curve={height=-12pt}, from=1-2, to=2-1]
    \end{tikzcd}\]
    which is a linear higher Dirac structure point-wise\footnote{In other words, a weak-Lagrangian subbundle (with no involutivity required)}. 
\end{Def}

\begin{theorem} Let $M$ be a manifold and $k \geq 1$ an integer. Then, there exists an injection
    $$\{\text{Almost graded Dirac structures on }M\} \hookrightarrow \{\text{Almost weak higher Dirac structures on }M\}$$
    which defines a bijection onto the set of all almost weak higher Dirac structures on $M$, $D$, such that
    $$D_p = \left\langle (V \wedge U, \iota_U \alpha) : U \in \bigvee_p M, (V, \alpha) \in D \right\rangle$$
    defines a vector subbundle, for each $1 \leq p \leq k$.
\end{theorem}

We also have the equivalence when restricting the above mapping to (involutive) graded Dirac structures.

\begin{theorem}
    \label{Dirac_structures:thm:Differential_equivalence}Let $M$ be a manifold and $k \geq 1$ an integer. Then, there exists an injection
    $$\{\text{Graded Dirac structures on }M\} \hookrightarrow \{\text{Weak higher Dirac structures on }M\}$$
    which defines a bijection onto the set of all almost weak higher Dirac structures on $M$, $D$, such that
    $$D_p = \left\langle (V \wedge U, \iota_U \alpha) : U \in \bigvee_p M, (V, \alpha) \in D \right\rangle$$
    defines a vector subbundle, for each $1 \leq p \leq k$.
\end{theorem}

\begin{proof} We only need to check that the involutivity of $D_1, \dots, D_k$ is determined by the involutivity of $D_1$. 
    By \cref{Linear_Graded_Dirac:corollary:Dp_in_terms_of_D1}, we know that we have the equality
    $$D_p = \langle (V \wedge U, \iota_U \alpha): U \in \bigvee_{p-1} M, \, (V, \alpha) \in D_1\rangle.$$ Now,
    since we have the equality $$\iota_U \beta = (-1)^{pq}\iota_V \alpha,$$ when restricted to a isotropic vector subbundle, then
    we can express the Courant bracket as
    $$\llbracket (U, \alpha), (V, \beta)\rrbracket = ([U,V], (-1)^{(p-1)q}\pounds_U \beta - \iota_V d \alpha).$$
    We will prove that the sections of $D_p$ are closed under the Courant bracket by induction.
    Suppose it closed for the indices $p, q - 1.$ We will prove it for the indices $p, q.$ By
    the description we gave in the linear case, it suffices to prove it for a pair of elements
    $$(U, \alpha) \in D_p, \, (V \wedge X, \iota_X \beta) \in D_q,$$
    where $(V, \beta) \in D_{q-1},$ and $X \in \mathfrak{X}(M).$ For this proof we will use the notation 
    $$\iota_X (U, \alpha) = (U\wedge X, \iota_X \alpha).$$ Notice that it defines a mapping 
    $$D_p \rightarrow D_{p-1}.$$ Now,
    \begin{align*}
        \llbracket(U, \alpha), (V \wedge X, \iota_X\beta) \rrbracket = 
        \left( [U, V \wedge X], (-1)^{(p-1)q} \pounds_U \iota_X \beta - \iota_{V \wedge X} \d{\alpha}\right)
    \end{align*}
    Since $$\pounds_U \iota_X \beta = (-1)^{p-1} \iota_{[U, X]} \beta + (-1)^{p-1} \iota_X \pounds_U \beta,$$
    we have
    \begin{align*}
        \llbracket(U, \alpha), (V \wedge X, \iota_X\beta) \rrbracket &= ([U, V] \wedge X + (-1)^{(p-1)(q-1)}V \wedge [U, X], \\
        &(-1)^{(p-1)(q-1)} \iota_{[U, X]} \beta + (-1)^{(p-1)(q-1)} \iota_X \pounds_U \beta - \iota_X(\iota_V \alpha) )\\
        &= \iota_X \llbracket(U, \alpha), (V, \beta) \rrbracket + \\
        &\left(  (-1)^{(p-1)(q-1)}V \wedge [U, X](-1)^{(p-1)(q-1)} \iota_{[U, X]} \beta\right) \\
        &= \iota_X \llbracket(U, \alpha), (V, \beta) \rrbracket + \iota_{[U,X]} (V, \beta),
    \end{align*}
    which takes values in $D_{p + q - 1}$ by induction hypothesis, proving integrability.
\end{proof}

\section{Graded Poisson structures}
\label{section:Graded_Poisson_structures}

As we mentioned in \cref{Overview:subsection:Weak_Higher_Dirac}, weak higher
Dirac structures were introduced to include higher Poisson structures. In \cref{section:Graded_Dirac_structures} we 
studied the graded analogue of higher Dirac manifolds. A natural question to ask is what would be the graded analogue
to higher Poisson structures, as defined in \cite{Bursztyn2015}. We define and study such structure in this section.


\subsection{Graded Poisson manifolds}
\label{Graded_Poisson_structures:subsection:Graded_Poisson_manifolds}

Recall that 
$$E_p = \bigvee_p M \oplus \bigwedge^{k+1-p} M.$$
One possible approach to define graded Poisson structures is the following:
\begin{enumerate}[i)]
    \item Take a higher Poisson structure $\sharp: S \rightarrow M$ of order $k$, $S \subseteq \bigwedge^k M.$
    \item It determines a weak higher Dirac structure defining 
    $$D:= \langle (\sharp(\alpha), \alpha): \alpha \in S\rangle.$$
    \item Under certain regularity conditions, using \cref{Dirac_structures:thm:Differential_equivalence},
    this weak higher Dirac structure determines an unique graded Dirac structure $D_p \subseteq E_p.$
    \item Define a graded Poisson structure as and object equivalent to the sequence of vector subbundles $D_p$, $1 \leq p \leq k$.
\end{enumerate}

The notion obtained through this procedure is the following:

\begin{Def}[Graded Poisson structure] A \textbf{graded Poisson structure} of degree $k$ on a manifold $M$ is a tuple
     $(S^a, \sharp_a,K_{k+1-a})$, $1 \leq a \leq k$, where $S^a$, $K_p$ are vector subbundles 
     \[\begin{tikzcd}
         {S^a} & {\bigwedge^aM} \\
         M
         \arrow[hook, from=1-1, to=1-2]
         \arrow[from=1-1, to=2-1]
         \arrow[curve={height=-12pt}, from=1-2, to=2-1]
     \end{tikzcd};\]
    \[\begin{tikzcd}
	    {K_p} & {\bigvee_pM} \\
	    M
	    \arrow[hook, from=1-1, to=1-2]
    	\arrow[from=1-1, to=2-1]
    	\arrow[curve={height=-12pt}, from=1-2, to=2-1]
    \end{tikzcd};\]

    and $\sharp_a$ are vector bundle mappings $$\sharp_a: S^a \rightarrow \bigvee_{k+1-a} M/K_{k+1-a}$$
    satisfying
    \begin{enumerate}[i)]
        \item $K_p = (S^a)^{\circ, p},$ for $p\leq a$ and $K_1 = 0$.
        \item The maps $\sharp_a$ are \textit{skew-symmetric}, that is, 
        $$\iota_{\sharp_a(\alpha)} \beta = (-1)^{(k+1-a)(k+1-b)}\iota_{\sharp_b(\beta)} \alpha,$$ for all 
        $\alpha \in S^a,$ $\beta \in S^b.$
        \item It is \textit{integrable:} For $\alpha: M \rightarrow S^a$, $\beta: M \rightarrow S^b$ sections such that $a + b \leq 2k+1$, and $U, V$ multivectors
        of order $p = k + 1 - a$, $q = k+1-b$, respectively such that $$\sharp_a(\alpha) = U + K_p, \,\, \sharp_b(\beta)= V + K_q,$$ we have
        that the $(a+ b - k)$-form 
        $$\theta := (-1)^{(p-1)q}\pounds_U \beta + (-1)^{q} \pounds_V \alpha - \frac{(-1)^q}{2} d \left( \iota_V \alpha + (-1)^{pq} \iota_U \beta\right)$$
        takes values in $S_{a+b-k},$ and $$\sharp_{a+b - k}(\theta) = [U, V] + K_{ p +q - 1}.$$
    \end{enumerate}
\end{Def}

\begin{remark} The vector subbundles $S^a$ determine a submodule of the module of $a$-forms $\Omega^a(M)$, defined as the space of all $a$-forms $\alpha \in \Omega^{a}(M)$ such that $\alpha |_{x} \in S^a,$ for every $x \in M.$ Throughout the rest of the text, we will say that $\alpha$ \textbf{takes values in} $S^a$ if the previous condition holds.
\end{remark}

A graded Poisson structure of order $k$ on $M$ determines a graded Dirac structure of order $k$ on $M$ as follows:

\begin{proposition} Let $(S^a, K_{k+1-a}, \sharp_a)$ be a graded Poisson structure on $M$. Then, its ``graph''
     $$D_p := \langle (U, \alpha) \in E_p: \sharp_{k+1-p}(\alpha) = U +K_{p} \rangle$$ is a graded Dirac structure.
\end{proposition}

\begin{proof} Checking that the family of vector subbundles $D_1, \dots, D_k$ is weakly-Lagrangian is exactly
    the proof of \cref{Dirac_Structures:Linear_Dirac_Structures:prop:EquivalentDescription_GradedDirac}, although in this case we have
    non-degeneracy $K_1 = 0$. Integrability follows directly from the definition. 
\end{proof}

Given a graded Poisson structure of order $k$, $(S^a, K_{k+1-a}, \sharp_a)$ on $M$, we can recover the structure in arbitrary degrees
from the structure in degree $a = k$, $\sharp_k: S^k \rightarrow TM$ as follows:

\begin{enumerate}[i)]
    \item $S^a = \left \langle\iota_U \alpha: U \in \bigvee_{k-a} M, \, \alpha \in S_k \right \rangle,$
    \item $K_p = (S_k)^{\circ,p},$
    \item $\sharp_a(\iota_U \alpha)  = \sharp_k(\alpha) \wedge U + K_{k+1-a},$ for $\alpha \in S_k,$ $U \in \bigvee_{k-a} M.$
\end{enumerate}

And we have

\begin{theorem}
    \label{Graded_Poisson_structures:thm:Equivalence}
    Let $M$ be a manifold and $k \geq 1$ an integer. Then, there exists an injection
    $$\{\text{Graded Poisson structures on }M\} \hookrightarrow \{\text{Higher Poisson structures on }M\}$$
    which defines a bijection onto the set of all higher Poisson structures on $M$, $(S, \sharp)$, such that both
    $$S^a = \left \langle\iota_U \alpha: U \in \bigvee_{k-a} M, \, \alpha \in S \right \rangle,\,\,  K_p = (S)^{\circ,p}$$
    define a vector subbundle, for each $1 \leq p \leq k,$ $1 \leq a \leq k$.
\end{theorem}

An important result in Poisson geometry is that a Poisson manifold is foliated by symplectic leaves, and that this
foliation determines completely the Poisson structure. Graded Poisson manifolds also admit a natural foliation:

\begin{Def}[Multisymplectic foliation] A \textbf{multisymplectic foliation} of order $k$ of $M$ is a possibly singular foliation $\mathcal{F}$ (in the sense of
Stefan-Sussmann \cite{Sussmann1973, Stefan1974}) together with a choice of a closed $(k+1)$-form $$\omega_F \in \Omega^{k+1}(F)$$ defined on each leaf of the foliation $F \in \mathcal{F},$ which is smooth in the following sense. For each choice of vector fields tangent to the foliation, $X_1, \dots, X_{k+1},$ the function defined by 
$$x \mapsto (\omega_F|_x)(X_1|_x, \dots, X_{k+1}|_x)$$ is smooth, where $F \in \mathcal{F}$ is the unique leaf passing through $x$, namely, $x \in F.$
\end{Def}

\begin{theorem}
    \label{Graded_Poisson_structures:thm:Poisson_implies_Multisymplecticfoliation}
    Let $(S^a, K_{k+1-a}, \sharp_a)$ be a graded Poisson structure on $M$. Then, $M$ admits a multisymplectic foliation.
\end{theorem}

\begin{proof} Define the (generalized) distribution $$E := \Im \sharp_k \subseteq TM.$$ Then, $E$ is integrable in the sense of
    Stefan and Sussmann (see \cite{Stefan1974,Sussmann1973}). Indeed, by integrability of the Poisson structure, we have
    that for $\alpha, \beta: M \rightarrow S_k$ sections, 
    $$[\sharp_k(\alpha), \sharp_k(\beta)] = \sharp_k(\theta),$$
    where
    $$\theta := (-1)^{(p-1)q}\pounds_{\sharp_a(\alpha)} \beta + (-1)^{q} \pounds_{\sharp_b(\beta)} \alpha - 
    \frac{(-1)^q}{2} \d {\left( \iota_{\sharp_b(\beta)} \alpha + (-1)^{pq} \iota_{\sharp_a(\alpha)} \beta\right)} \in S,$$
    concluding that $E$ is closed under Lie bracket. Therefore, $E$ determines a foliation on $M$, $\mathcal{F}.$ 
    Let $F \in \mathcal{F}$ be a leaf and $x \in F.$  Then, by definition, $T_x F = E |_x.$ Define
    $$\flat: E|_x \rightarrow \bigwedge^{k}(E|_x)^\ast$$
    as $$\flat(\sharp_k(\alpha)) := i^\ast \alpha,$$ where $\alpha \in (S_k)|_x$, and $$i: E\hookrightarrow TM$$
    denotes the inclusion. \\

    Let us check that $\flat$ is well-defined. Indeed, suppose $$\sharp_k(\alpha) = \sharp_k(\beta).$$ 
    Then, for each $\gamma_1, \dots, \gamma_k \in S^k$, 
    \begin{align*}
        \alpha (\sharp_k(\gamma_1), \dots, \sharp_k(\gamma_k)) &= - \gamma_1(\sharp_k(\alpha), \dots, \sharp_k(\gamma_k))\\
        &= - \gamma_1(\sharp_k(\beta), \dots, \sharp_k(\gamma_k))\\
        &= \beta (\sharp_k(\gamma_1), \dots, \sharp_k(\gamma_k)),
    \end{align*}
    which yields $i^\ast \alpha = i^\ast \beta$ and proves well-definedness.\\

    Now, since for each $\alpha, \beta \in S^k$ we have
    \begin{align*}
        \iota_{\sharp_k(\alpha)} \beta = - \iota_{\sharp_k(\beta)} \alpha,
    \end{align*}
    we have that for each $u, v \in E|_x$, $$\iota_u \flat(v) = - \iota_v \flat(u).$$
    Therefore, $\flat$ determines a $(k+1)$-form $$\omega |_x \in \bigwedge^{k+1} (E|_x)^\ast.$$ This defines a $(k+1)$-form
    on each leaf $F \in \mathcal{F},$ which we shall denote $\omega_F.$ The collection of these forms is smooth. The proof is finished once we show that $\d{\omega_F} = 0.$\\

    Let $X, Y$ be vector fields on $F$, for certain leaf $F \in \mathcal{F}.$ Then, at least locally, there exists forms $\alpha$, $\beta$
    on $M$ taking values in $S^k$, such that $X$, $\sharp_k(\alpha),$ and $Y, \sharp_k(\beta)$ are $\iota$-related,
    where $$\iota: F \rightarrow M$$ denotes the inclusion, that is, such that for every $x \in F$,
    $$ X |_x = \sharp_k(\alpha)|_x , \,  Y|_x = \sharp_k(\beta)|_x.$$
    Then $$[X, Y]|_x = [\sharp_k(\alpha), \sharp_k(\beta)] |_x,$$ which, by involutivity of the Poisson structure,
    gives 
    \begin{align*} 
        [X, Y]|_x& = [\sharp_k(\alpha), \sharp_k(\beta)] |_x \\
        &= \sharp_k(\theta |_x) = \sharp_k\left( \pounds_{\sharp_k(\alpha)} \beta - \pounds_{\sharp_k(\beta)} 
        \alpha + \frac{1}{2}\d{(\iota_{\sharp_{k}(\beta)} \alpha - \iota_{\sharp_k(\alpha)}\beta)} \right)\big |_x\\
        &= \sharp_k\left ( \pounds_{\sharp_k(\alpha)} \beta - \iota_{\sharp_k(\beta)} \d{\alpha}\right ) \big |_x.
    \end{align*}
    Recall that
    \begin{align} 
        \iota_{[X,Y]} \omega_F &= \pounds_X \iota_Y \omega_F - \iota_Y \pounds_X \omega_F\\
        \label{Graded_Poisson_structures:eq:contraction_bracket}
        &= \pounds_X \iota_Y \omega_F - \iota_Y d {\iota_X \omega_F} - \iota_Y \iota_X \d{\omega_F}.
    \end{align}
    Now, since $$\iota_{\sharp_k(\alpha)} (\omega_F)|_x = \iota^\ast \alpha,$$
    for every $\alpha \in (S^k) \big |_x,$ we have, on the one hand
    \begin{align*}
        \left(\iota_{[X, Y]} \omega_F \right)|_x &= \iota_{\sharp_k(\theta)|_x} (\omega_F)|_x = i^\ast(\theta)|_x\\
        &= \iota^\ast\left( \pounds_{\sharp_k(\alpha)} \beta - \iota_{\sharp_k(\beta)} \d{\alpha}\right)\big |_x,
    \end{align*}
    and, on the other hand,
    \begin{align*}
        \left(\pounds_X \iota_Y \omega_F - \iota_Y \d{\iota_X \omega_F} - \iota_Y \iota_X \d{\omega_F} \right)|_x = 
        \iota^\ast\left(\pounds_{\sharp_k(\alpha)} \sharp_k(\beta) - \iota_{\sharp_k(\beta)} \d{\sharp_k (\alpha)}\right)|_x
        - \left(\iota_Y \iota_X \d{\omega_F} \right)|_x.
    \end{align*}
    Using these last two equalities and \cref{Graded_Poisson_structures:eq:contraction_bracket}, we must have
    $$\left(\iota_Y \iota_X \d{\omega_F} \right)|_x = 0.$$ Since $X, Y$ and $x \in F$ are arbitrary, it follows that $$\d{\omega_F} = 0,$$
    concluding the proof.
\end{proof}

\begin{remark} In the case $k = 1$ this recovers the classical symplectic foliation of Poisson manifolds. 
Indeed, in this case the induced closed $2$-form on each leaf $\omega_F,$ $F \in \mathcal{F}$ is non-degenerate, giving a symplectic form.
\end{remark}

For the converse, we can recover a graded Poisson structure from non-degenerate multisymplectic foliations with some regularity condition. The necessity of these regularity conditions is to be expected. Indeed, in Poisson geometry, the Poisson structure is completely determined by its symplectic foliation, but an arbitrary symplectic foliation does not necessarily arise from a bivector field, needing differentiability conditions in order to guarantee the converse (see \cite{Vaisman1994}). In a graded setting, it will not be true that a graded Poisson structure is characterized by its induced multisymplectic foliation. Nevertheless, under certain regularity conditions, we can guarantee that the induced foliation defines an extension of the original graded Poisson structure:

\begin{theorem}
    \label{Graded_Poisson_structures:thm:Multisymplecticfoliation_implies_Poisson}
    Let $(\mathcal{F}, \omega_\mathcal{F})$ be a non-degenerate multisymplectic foliation of $M$, and suppose that the collection of subspaces 
    $$S |_x := (i^\ast)^{-1}\left( \Im \flat_{\omega_F} |_x\right)$$ defines a vector subbundle $$S \hookrightarrow \bigwedge^k M,$$ where $\flat_{\omega_F}: T_x F \rightarrow \bigwedge^k T^\ast_x F$ denotes the morphism induced by $\omega_{F}$, and $$i^\ast: \bigwedge^k T^\ast_x M \rightarrow T^\ast_x F$$ denotes the restriction of forms. Suppose further that the collection of mappings $$\flat^{-1}:\Im \flat_{\omega_F} \rightarrow TF, F \in \mathcal{F}$$ is smooth in the following sense: $\flat^{-1}(i^\ast \alpha) \in \mathfrak{X}(M)$ when $\alpha \in \Omega^k(M)$ takes values in $S$. Then,
    $M$ admits a higher Poisson structure $\sharp: S \rightarrow TM$ such that $(\mathcal{F}, \omega_\mathcal{F})$ is the multisymplectic foliation
    obtained through \cref{Graded_Poisson_structures:thm:Poisson_implies_Multisymplecticfoliation}.
\end{theorem}

\begin{proof} It is enough to define $$\sharp: S |_x \rightarrow T_xM$$ as $$\sharp(\alpha) := e,$$ if $\flat_{\omega_F}(e) = i^\ast \alpha.$ This is clearly well-defined (by non-degeneracy of $\flat_{\omega_F}$) and it is clearly skew-symmetric. Smoothness of $\sharp$ is clear from the hypotheses. Involutivity follows from a similar argument to the one made in \cref{Graded_Poisson_structures:thm:Poisson_implies_Multisymplecticfoliation}.
\end{proof}

\begin{remark} Again, the case $k = 1$ recovers the fact that a Poisson manifold is determined by its symplectic foliation. Indeed, in this case $$S = T^\ast M,$$ and the smoothness condition is implied by the fact that $\sharp: T^\ast M \rightarrow TM$ is smooth. However, for degree $k > 1$ we do not recover the original Poisson structure from its foliation, even in the hypotheses of \cref{Graded_Poisson_structures:thm:Multisymplecticfoliation_implies_Poisson}. Nevertheless, we do recover an extension of it when it defines a non-degenerate multisymplectic foliation, and $S$ (the collection of spaces defined in \cref{Graded_Poisson_structures:thm:Multisymplecticfoliation_implies_Poisson}) is a vector subbundle. Indeed, given a higher Poisson structure of order $k$ $$\sharp: S_k \rightarrow TM$$ satisfying the previous conditions, smoothness of $\flat^{-1}$ follows easily from smoothness of $\sharp.$ Furthermore, we clearly have $S_k \subseteq S,$ and that the obtained morphism through \cref{Graded_Poisson_structures:thm:Multisymplecticfoliation_implies_Poisson} extends the original one.
\end{remark}

\begin{remark} In general, we cannot extend the vector subbundle defined in \cref{Graded_Poisson_structures:thm:Multisymplecticfoliation_implies_Poisson} to a graded Poisson structure. However, 
    if we can guarantee that the obtained vector subbundle $S \subseteq \bigwedge^k M$ satisfies the hypothesis of 
    \cref{Graded_Poisson_structures:thm:Equivalence}, we obtain a graded Poisson structure whose 
    multisymplectic foliation is the original one.
\end{remark}

We have a series of interesting particular 
cases of \cref{Graded_Poisson_structures:thm:Multisymplecticfoliation_implies_Poisson}:



\begin{Def} A multisymplectic foliation $\mathcal{F}, \omega_F, F \in \mathcal{F}$ is said to have \textbf{constant linear type} if for any pair of leaves, $ F_1, F_2 \in \mathcal{F}$ and points $x_1 \in F_1, x_2 \in F_2$ there exists a linear isomorphism
$$\phi: T_{x_1} F_1 \rightarrow T_{x_2} F_2$$ satisfying 
$$\phi^\ast (\omega_{F_2})|_{x_2} = (\omega_{F_1})|_{x_1}.$$
\end{Def}

\begin{remark} If a multisymplectic foliation is of constant linear type, it is of constant dimension.
\end{remark}

\begin{corollary}
\label{constant-linear-type}
Non-degenerate multisymplectic foliations of constant linear type are induced by graded Poisson structures
\end{corollary}
\begin{proof} Since $\mathcal{F}$ is of constant dimension, it is integrable in the sense of Frobenius, and thus, locally, there exist coordinates $(x^1, \dots, x^k, y^1, \dots, y^m)$ around any point such that $$T_x F = \left \langle \pdv{x^1} \bigg |_x, \dots, \pdv{x^k} \bigg |_x \right \rangle,$$ where $x \in F.$ In particular, the multisymplectic form (on each leaf) may be written as $$\omega_F |_x  = f_{i_1 \dots i_{k+1}}(x) \d{x}^{i_1} \wedge \cdots \wedge \d{x}^{i_{k+1}} |_x,$$ $x \in F,$ for certain coefficients $f_{i_1, \dots, i_{k+1}}(x) \in \mathbb{R}.$ Furthermore, since the multisymplectic foliation is smooth, these coefficients define smooth functions on $M$. It is clear that if 
$$S = (i^\ast)^{-1} \Im \flat_{\omega_F},$$ where $$i^\ast_x: \bigwedge^k T_x^\ast M \rightarrow \bigwedge^k T^\ast_x F$$ is the restriction and
$$\flat_{\omega_F}: T_xF \rightarrow \bigwedge^k T^\ast_x F$$ is the contraction by $\omega_F$, then $S$ it is locally generated by
$$S = \langle f_{l i_2, \dots, i_{k+1}} \d{x}^{i_2} \wedge \cdots \wedge \d{x}^{i_{k+1}}, l = 1, \dots, k\rangle \oplus \langle \d{x^{i_1}} \wedge \cdots \wedge \d{x}^{i_l}\wedge \d{y}^{j_{l+1}} \wedge \cdots \d{y}^{j_{k}}, i_1 < \cdots i_l, j_{l+1} < j_{k}\rangle,$$ which is smoothly generated and of constant rank. As a consequence, it defines a vector subbundle. Furthermore, it may be readily observed that the family of mappings $$\flat^{-1}: \Im \flat_{\omega_F} \rightarrow TF, F \in \mathcal{F}$$ is smooth. Using \cref{Graded_Poisson_structures:thm:Multisymplecticfoliation_implies_Poisson}, we get a higher Poisson structure $$\sharp: S \rightarrow TM$$ whose multisymplectic foliation is the given one.\\

Finally, since the linear type is constant,
    we are in the hypothesis of \cref{Graded_Poisson_structures:thm:Equivalence}.
    Indeed, given a pair of points $x_1, x_2 \in M$ and denoting by $F_1, F_2$ the corresponding leaves containing $x_1$ and $x_2$, respectively, we have a linear isomorphism $$\phi: T_{x_1} F_1 \rightarrow T_{x_2} F_2$$ satisfying $$\phi^\ast (\omega_{F_2})|_{x_2} = (\omega_{F_1})|_{x_1}.$$
    Extend this isomorphism to a linear isomorphism 
    $$\psi: T_{x_1} M \rightarrow T_{x_2} M.$$
    As a quick check shows, $$S^k \big |_{x_1} = \psi^\ast \left( S^k \big |_{x_2} \right),$$
    where $S^k$ is the subspace of forms defined in 
\cref{Graded_Poisson_structures:thm:Multisymplecticfoliation_implies_Poisson}.
    Then, we can define an isomorphism between the subspaces
    $$S^a \big |_{x_1} = \langle (U, \iota_U \alpha), U \in \bigvee_{k-a} M \big |_{x_1}, \alpha \in S^k \big |_{x_1}\rangle;$$ $$S^a \big |_{x_2} = \langle (U, \iota_U \alpha), U \in \bigvee_{k-a} M \big |_{x_2}, \alpha \in S^k \big |_{x_2}\rangle $$ given by 
    $$(U, \iota_U \alpha) \mapsto (\psi_\ast U, \iota_{\psi_\ast U} (\psi^{-1})^\ast \alpha).$$
    We conclude that 
    $$\dim S^a \big |_{x_1} = \dim S^a \big |_{x_2},$$ and therefore, $S^a$ is a vector subbundle, for each $1 \leq a \leq k$, proving that we get a graded Poisson structure.

\end{proof}

\begin{corollary} 
\label{corollary:volume_forms}
Manifolds foliated by leaves which are $(k+1)-$dimensional orientable manifolds
    together with a smoothly varying family of volume forms,
    admit a graded Poisson structure whose multisymplectic foliation is the given one.
\end{corollary}

\begin{proof} We are clearly in the situation of \cref{constant-linear-type}.
\end{proof}

\begin{remark}
Unlike the case of Poisson Geometry, a multisymplectic foliation does not uniquely determine a higher (nor
graded) Poisson structure. Indeed, there could be several available choices for the subbundle $S$. 
\end{remark}

 The results presented in this section may be easily generalized to weak higher (and graded) Dirac structures. In this case, the integrable distribution may be taken as $$E:= \operatorname{pr}_1(D_1),$$ and the corresponding $(k+1)$-form $\omega_x \in \bigwedge^{k+1} (E|_x)^\ast$ is 
$$(\omega|_x)(e_1, \dots, e_{k+1}) := \alpha(e_2, \dots, e_{k+1}),$$ where $\alpha \in \bigwedge^k T^\ast_x M$ is 
a $k$-form such that
$$(e_1, \alpha) \in D_1.$$ Conversely, given a pre-multisymplectic foliation of $M$, $\mathcal{F}, \omega_F$, the natural choice for the corresponding weak higher Dirac structure would be the following:
$$D|_x := \{(e, \alpha): \iota^\ast \alpha = \flat_{\omega_F}(e), e \in E|_x \, \text{and } \ker \flat_{\omega_F} \subseteq \langle \alpha\rangle ^{\circ, 1}\}\footnote{
The reason for the last condition is to obtain a weak Lagrangian subbundle}.$$ 

Then, we obtain

\begin{theorem}
\label{Graded_Poisson_structures:thm:Multisymplecticfoliation_implies_Dirac}
Let $(\mathcal{F}, \omega_F)$ be a multisymplectic foliation of $M$ satisfying the hypothesis above, and suppose that the collection of vector subspaces defined above $D|_x$ defines a vector subbundle $$D \subseteq TM \oplus_M \bigwedge^k M.$$ Then, it is a weak higher Dirac structure whose multisymplectic foliation is the one given.
\end{theorem}

\begin{proof} It is clearly isotropic. Now we only need to check that it is weakly Lagrangian or, equivalently, that
$$D \cap TM = \ker \flat_{\omega_F} = \left(\{\alpha \in \ker \flat_{\omega_F}^{\circ, 1}: i^\ast \alpha \in \Im \flat_{\omega_F}\} \right)^{\circ, 1} = (\operatorname{pr}_2 D)^{\circ, 1},$$ which is clear, as a quick argument shows. Integrability is easily obtained from closedness of each of the multisymplectic forms.
\end{proof}

\begin{remark} This results recovers as a particular case that a Dirac manifold is completely determined by its induced presymplectic foliation. As the reader might expect from the discussion of graded Poisson manifolds, this is not true for $k > 1.$ However, if a given weak higher Dirac structure $D$ satisfies
$$\ker \flat_{\omega_F} = D \cap TM\footnote{Notice that this corresponds to the fact that $\omega$ defines a non-singular multisymplectic form when $D$ defines a higher Poisson structure},$$ 
then the Dirac structure obtained through \cref{Graded_Poisson_structures:thm:Multisymplecticfoliation_implies_Dirac}, which we denote by $\widetilde D$, defines an extension of the original, in the sense that $$D \subseteq \widetilde D,$$ as one may easily check.
\end{remark}

\begin{corollary} Let $(\mathcal{F}, \omega_F)$ be a multisymplectic foliation of constant linear type. Then, there exists a graded Dirac structure on $M$ such that $(\mathcal{F}, \omega_F)$ is the induced multisymplectic foliation.
\end{corollary}

\begin{proof} It follows from a similar argument to the one presented in \cref{constant-linear-type}.
\end{proof}

\begin{remark} The foliated nature of these structures was already studied in \cite{Bursztyn2019}, where the authors gave a characterization of weak higher Dirac structures in terms of a foliation together with a tensorial object (a cocycle of certain differential complex). Nevertheless, our approach allows for the construction of non trivial examples, and the analogy with the classical results is readily seen.
\end{remark}

\subsection{Graded Poisson brackets}
\label{Graded_Poisson_structures:subsection:Graded_Poisson_brackets}

So far, we have studied the tensorial aspect of graded (and higher) Poisson structures. Poisson structures ($k =1$),
that is, bivector fields $\Lambda \in \mathfrak{X}^2(M)$ satisfying $[\Lambda, \Lambda] = 0$, are characterized
by the induced Poisson bracket $\{ \cdot, \cdot \}$  on the space of functions $C^\infty(M).$ This is a fundamental property because,
in terms of mechanics, it allows us to study the dynamics just using the induced algebra, {a fundamental property to study quantization of the system}.\\

In classical field theories, there have been several attempts to identify the corresponding structure that observables
(i.e. forms) have. For instance, we have $L_\infty$-algebras (see \cite{Baez2010,Ryvkin2016,Rogers2012}), by studying the algebraic
structure of Hamiltonian $(k-1)$-forms; or graded Lie algebras (see \cite{Cantrijn1996,Vankerschaver2011,Leon2024a}), by studying 
the structure of general Hamiltonian forms when quotiented by exact forms. However, these structures do not recover the corresponding multisymplectic structure 
(if they are induced by a multisymplectic form), nor they recover the graded Poisson structure (if they are induced by a graded
Poisson structure). In this section we identify the algebraic structure that a graded Poisson structure induces,
and prove that, under some integrability conditions on the vector subbundles $S^a$, it completely characterizes the geometry.\\

\begin{Def}[Hamiltonian form] Let $(S^a, K_p, \sharp_a)$ be a graded Poisson structure of order $k$ 
    on $M$. An $(a-1)$-form $ \alpha \in \Omega^{a-1}(M)$
    is called \textbf{Hamiltonian} if $\d{\alpha}$ takes values in $S^a$. Denote by $\Omega^{a-1}_H(M)$ the space of all Hamiltonian forms.
\end{Def}

\begin{Def}[Poisson bracket of Hamiltonian forms] Given $\alpha \in \Omega^{a-1}_H(M)$, $\beta \in \Omega^{b-1}_H(M)$, define
    their \textbf{Poisson bracket} as the form
    $$\{\alpha, \beta\} := (-1)^{\deg \beta}\iota_{\sharp_b(\d{\beta})} \d \alpha$$
\end{Def}


\begin{proposition} Given two Hamiltonian forms, $\alpha \in \Omega^{a-1}_H(M)$, $\beta \in \Omega^{b-1}_H(M),$ their Poisson bracket is again 
    Hamiltonian, that is, $\{\alpha, \beta\} \in \Omega_H^{a + b - (k+1)}(M)$
\end{proposition}
\begin{proof} Indeed, by the integrability property of the graded Poisson structure, we know that
    $$\pounds_{\sharp_a(\d{\alpha})} \d{\beta} - \iota_{\sharp_b(\d{\beta})} \d{\d{\alpha}} = \d{\iota_{\sharp_a(\d{\alpha})}} \d{\beta}$$
    takes values in $S^{a + b - k}$. Therefore, $\{\alpha, \beta\}$ is Hamiltonian.
\end{proof}

Then, we get a well defined bracket on the space of Hamiltonian forms
$$\Omega^{a-1}_H(M) \otimes \Omega^{b-1}_H(M) \rightarrow \Omega^{a + b - (k+1)}_H(M).$$

Our objective now is studying the properties of $\{\cdot, \cdot\}$. In particular, we will obtain that this bracket defines a graded Lie algebra (modulo exact forms), for a suitable notion of degree.

\begin{Def} Given a Hamiltonian $(a-1)-$form $\alpha \in \Omega^{a-1}_H(M)$, define its degree as $$\deg \alpha := k - a 
    = k - 1 - \operatorname{ord} \alpha.$$ 
\end{Def}

\begin{proposition}
    \label{Graded_Poisson_bracket:prop:Hamiltonian_multivector_of_bracket}
    Let $\alpha, \beta$ be arbitrary Hamiltonian forms. If
    $$\sharp_a(\d{\alpha}) = U + K_{k+1-a}, \, \sharp_b(\d{\beta}) = V + K_{k+1-b},$$ then
    $$\sharp_{a + b - (k+1)} (\d{\{\alpha, \beta\}}) = -[U,V ] + K_{2k +1 - (a  + b)}$$
\end{proposition}
\begin{proof} If follows easily from the integrability of the Poisson structure. Indeed, defining 
    $p:= k +1 -a, q := k +1 - b$, we know that the form
    \begin{align*}
        \theta &:= (-1)^{(p-1)q}\pounds_U \d{\beta} + (-1)^{q} \pounds_V \d{\alpha}
         - \frac{(-1)^q}{2} \d{ \left( \iota_V\d{\alpha} + (-1)^{pq} \iota_U \d{\beta}\right)}\\
        &= (-1)^{(p-1)} \pounds_U \d{\beta} = (-1)^{(p-1)q} \d{ \iota_U} \d{\beta}\\
        &= (-1)^q \d{\iota}_{\sharp_a(\d{\alpha})} \d{\beta} = (-1)^q\d{\iota}_{\sharp_b(\d{\beta})} \d{\alpha} = - \{\alpha, \beta\}
    \end{align*}
    satisfies
    $$\sharp_{a + b - (k+1)}(\theta) = [U, V] + K_{p+q-1},$$
    which finishes the proof.
\end{proof}

\begin{lemma}
    \label{Graded_Poisson_Brackets:lemma:Jacobi}
    Let $\alpha \in \Omega^{a-1}_H(M) , \beta \in \Omega^{b-1}_H(M), \gamma \in \Omega^{c-1}_H(M)$ be Hamiltonian forms. Then (omitting
    the indices of $\sharp_a$), 
    $$(-1)^{(\deg \gamma - 1) \deg \alpha} \iota_{\sharp(\d{\gamma})} \d{\iota}_{\sharp(\d{\beta})} \d{\alpha} + \text{clyc.} = 
    (-1)^{(\deg \alpha + \deg \gamma) (\deg \beta + 1)} \d{\iota_{\sharp(\d{\alpha})}} \iota_{\sharp(\d{\beta})} \d \gamma.$$ 
\end{lemma}
\begin{proof}Denote by $p, q, r$ the order of the multivector fields $\sharp(\d{\alpha}), \sharp(\d{\beta}), \sharp(\d{\gamma}),$
    respectively. We have $$p = \deg \alpha + 1, q = \deg \beta + 1, r = \deg \gamma + 1.$$
    Now,
    \begin{align*}
        (-1)^{(\deg\gamma - 1) \deg \alpha}\iota_{\sharp(\d{\gamma})} \d{\iota}_{\sharp(\d{\beta})} d \alpha& = 
        (-1)^{r (p - 1)}\iota_{\sharp(\d{\gamma})} \d{\iota}_{\sharp(\d{\beta})} d \alpha = \\
        &(-1)^{r (p - 1) + r(p +q - 1)} \iota_{\sharp(\d{\iota}_{\sharp(\d{\beta})} d \alpha)} \d \gamma.
    \end{align*}
    From \cref{Graded_Poisson_bracket:prop:Hamiltonian_multivector_of_bracket},
    $$\sharp(\d{\iota}_{\sharp(\d{\beta})} d \alpha) = (-1)^q [\sharp(\d{\alpha}), \sharp(\d{\beta})],$$ and therefore
    \begin{align*}
        &(-1)^{r (p - 1) + r(p +q - 1)} \iota_{\sharp(\d{\iota}_{\sharp(\d{\beta})} d \alpha)} \d \gamma=
        (-1)^{rq +q}\iota_{ [\sharp(\d{\alpha}), \sharp(\d{\beta})]} \d \gamma \\
        &= (-1)^{rq +q} \left ( (-1)^{(p-1)q} \pounds_{\sharp(\d \alpha)} \iota_{\sharp(\d \beta)} \d \gamma 
        - \iota_{\sharp(\d \beta)} \pounds_{\sharp(\d{\alpha})} \d \gamma\right )\\
        & = (-1)^{rq + pq} \d{\iota_{\sharp(\d \alpha)}} \iota_{\sharp(\d \beta)} \d \gamma - 
        (-1)^{rq + pq+ p} {\iota_{\sharp(\d \alpha)}} \d{\iota}_{\sharp(\d \beta)} \d \gamma 
        - (-1)^{rq + q} {\iota_{\sharp(\d \beta)}} \d{\iota}_{\sharp(\d \alpha)} \d \gamma.
    \end{align*}
    Interchange $\beta$ and $\gamma$ using the properties of $\sharp$ to get
    \begin{align*}
        &(-1)^{(\deg\gamma - 1) \deg \alpha}\iota_{\sharp(\d{\gamma})} \d{\iota}_{\sharp(\d{\beta})} d \alpha\\
        &= (-1)^{rq + pq} \d{\iota_{\sharp(\d \alpha)}} \iota_{\sharp(\d \beta)} \d \gamma - 
        (-1)^{ pq+ p} {\iota_{\sharp(\d \alpha)}} \d{\iota}_{\sharp(\d \gamma)} \d \beta 
        - (-1)^{rq + q} {\iota_{\sharp(\d \beta)}} \d{\iota}_{\sharp(\d \alpha)} \d \gamma\\
        &= (-1)^{(\deg \alpha + \deg \gamma)(\deg \beta + 1)}\d{\iota_{\sharp(\d \alpha)}} \iota_{\sharp(\d \beta)} \d \gamma -
        (-1)^{ (\deg \alpha - 1) \deg \beta} {\iota_{\sharp(\d \alpha)}} \d{\iota}_{\sharp(\d \gamma)} \d \beta\\
        & - (-1)^{(\deg \beta - 1) \deg \gamma} {\iota_{\sharp(\d \beta)}} \d{\iota}_{\sharp(\d \alpha)} \d \gamma ,
    \end{align*}
    proving the equality.
 \end{proof}

The main properties of the Poisson bracket $\{ \cdot, \cdot \}$ are the following

\begin{theorem}[Properties of the Poisson bracket] 
    \label{Graded_Poisson_structures:thm:Properties_Poisson_Bracket}
    Let $(S^a, K_p, \sharp_a)$ be a graded Poisson structure on a manifold
    $M$. Then, for arbitrary $\alpha \in \Omega^{a-1}_H(M), \beta \in \Omega^{b-1}_H(M), \gamma \in \Omega^{c-1}_H(M)$, 
    the Poisson bracket $\{\cdot, \cdot\}$ satisfies
    \begin{enumerate}[i)]
        \item \textit{It is graded:} $$\deg \{\alpha, \beta\} = \deg \alpha + \deg \beta;$$
        \item \textit{It is graded-skew-symmetric:} $$\{\alpha, \beta\} = -(-1)^{\deg \alpha \deg \beta} \{\beta,\alpha\};$$
        \item \textit{It is local:} If $\d{\alpha} |_x = 0,$ $\{\alpha, \beta\}|_x = 0$
        \item \textit{It satisfies Leibniz identity:} For $a = k$, if $\beta \wedge \d{\gamma} \in \Omega^{b + c - 1}_H(M),$ then 
        $$\{ \beta \wedge \d{\gamma}, \alpha\} = \{\beta, \alpha\} \wedge \d{\gamma} + (-1)^{k -\deg \beta} \d{\beta} \wedge \{\gamma, \alpha\};$$
        \item \textit{It is invariant by symmetries:} If $X \in \mathfrak{X}(M)$ and $\pounds_X \alpha = 0,$ then 
        $\iota_X \alpha \in \Omega^{a-2}_H(M)$ and
        $$\{\iota_X \alpha, \beta\}=  (-1)^{\deg \beta} \iota_X \{\alpha, \beta\};$$
        \item \textit{It satisfies graded Jacobi identity (up to an exact term):} 
        $$(-1)^{\deg \alpha \deg \gamma}\{ \{\alpha, \beta \} , \gamma\} + \text{cyclic terms} = \text{exact form}.$$
    \end{enumerate}
\end{theorem}

\begin{proof}
    \begin{enumerate}[i)]
        \item Indeed, $\{\alpha, \beta\}$ is a $((a + b ) - (k+1))-$form and so,
        $$\deg \{\alpha, \beta\} = 2 k - a - b = \deg \alpha + \deg \beta.$$
        \item 
        \begin{align*}
            \{\alpha, \beta\} &= (-1)^{\deg \beta} \iota_{\sharp_b(\d{\beta})}
            \d\alpha = (-1)^{\deg \beta + (\deg \alpha +1) (\deg \beta + 1)} \iota_{\sharp_a(\d{\alpha})}\d\beta\\
            &= (-1)^{\deg \alpha \deg \beta + \deg \alpha  + 1} = - (-1)^{\deg \alpha \deg \beta}\{\beta , \alpha\}.
        \end{align*}
        \item It is immediate.
        \item 
        \begin{align*}
            \{\beta \wedge \d{ \gamma}, \alpha\}& = (-1)^{\deg \alpha} \iota_{\sharp_k(\d{\alpha})}
            \left( \d{\beta} \wedge \d{\gamma}\right) \\ &= (-1)^{\deg \alpha} \iota_{\sharp_k(\d{\alpha})} \d{\beta} \wedge \d{\gamma}
            + (-1)^{\deg  \alpha + b}\d{\beta} \wedge\iota_{\sharp_k(\d{\alpha})}\d{\gamma} \\
            &=  \{\beta, \alpha\} \wedge \d{\gamma} + (-1)^{b} \d{\beta} \wedge \{\gamma, \alpha\}\\
            &= \{\beta, \alpha\} \wedge \d{\gamma} + (-1)^{k - \deg \beta} \d{\beta} \wedge \{\gamma, \alpha\}
        \end{align*}
        \item
        \begin{align*}
            \{\iota_X \alpha, \beta\} &= (-1)^{\deg \beta} \iota_{\sharp_b(\d{\beta})} \d{\left(\iota_X \alpha \right)} \\
            &= (-1)^{\deg \beta + 1} \iota_{\sharp_b( \d {\beta})} \iota_X \d{\alpha} 
            = \iota_X \iota_{\sharp_b( \d {\beta})}\d{\alpha}  \\
            &= (-1)^{\deg \beta} \iota_X\{\alpha, \beta\}
        \end{align*}
        \item Using \cref{Graded_Poisson_Brackets:lemma:Jacobi} we get
        $$(-1)^{\deg \gamma \deg \alpha + \deg \alpha + \deg \beta + \deg \gamma} \{\{\alpha, \beta\}, \gamma\} + \text{cylc.} = \text{exact form},$$
        and we obtain graded Jacobi identity once we multiply by $(-1)^{\deg \alpha + \deg \beta + \deg \gamma},$ which is a cyclic term.
    \end{enumerate}
\end{proof}
\begin{remark} A different approach for graded Poisson brackets is taken in \cite{Grabowski1997}, where J. Grabowski defines an extension of the Poisson bracket in a Poisson manifold to arbitrary differential forms, having the following graded nature $$\order \{\alpha, \beta\} = \order \alpha + \order \beta$$. In our case, the graded nature is modified to $$\order\{\alpha, \beta\} = \order \alpha + \order \beta - (k-1).$$
\end{remark}
\begin{corollary}
\label{Corollary:oriented_leaves_algebra}
Any manifold foliated by orientable $(k+1)-$dimensional leaves, together with a smoothly varying family of volume forms admits an algebraic invariant, namely, the graded Poisson bracket associated with the graded Poisson structure obtained through \cref{corollary:volume_forms}.
\end{corollary}

 
\begin{obs}
    The previous corollary is a generalization of a result proved by M. Zambon in \cite{Zambon2012}, associating to any compact orientable manifold an $L_\infty$-algebra. This $L_\infty$ algebra is the restriction of the Poisson algebra to $(\dim M - 2)-$forms induced by the non-degenerate multisymplectic structure defined by any volume form orienting the manifold (two such forms induce the same algebra up to isomorphism). Notice that the assumption on the existence of a family of volume forms varying smoothly is fundamental. Indeed, in general, orientable leaves may not admit a continuously varying orientation, e.g. the Möbius band foliated by lines.
\end{obs}

Now, we can ask whether a bracket satisfying the previous properties characterizes the graded Poisson structure.\\

The naive approach would be the following:
\begin{enumerate}[i)]
    \item
    \label{Conditions_on_S_a}
    Fix a sequence of vector subbundles of forms $S^a \subseteq \bigwedge^a M,$ for $1 \leq a \leq k$ 
    satisfying $$(S^a)^{\circ, p} = (S^b)^{\circ, p}, \, (S^k)^{\circ, 1} = 0,$$ for $a, b \geq p.$
    \item Define the set of Hamiltonian forms as $$\Omega^{a}_H(M): = \left \langle \alpha \in \Omega^{a-1}(M) : \, \d{\alpha} \in S^a
    \right \rangle.$$
    \item Define a Poisson bracket of order $k$ as a bilinear operation 
    $$\Omega^{a-1}_H(M) \otimes \Omega^{b-1}_H(M) \rightarrow \Omega^{a + b - (k+1)}_H(M),$$
    satisfying all the properties of \cref{Graded_Poisson_structures:thm:Properties_Poisson_Bracket}.
    \item Check whether this bracket is induced by an unique graded Poisson structure 
    $$\sharp_a : S^a \rightarrow \bigvee_{k+1-a}M /K_{k+1-a}.$$
\end{enumerate}

However, this approach presents some technical difficulties. Indeed, in general, fixed the family of vector
subbundles $S^a,$ $1 \leq a \leq k,$ satisfying property \ref{Conditions_on_S_a}, the set of Hamiltonian forms may be trivial in some degrees.
One possible attempt to fix this first problem would be to restrict to locally defined Hamiltonian forms. This does
not get us very far either, since there may not exist closed forms taking values in $S^a$. Indeed, suppose $S^a$
locally generated as a vector subbundle by
$$S^a \big |_x = \langle \alpha_1 |_x, \dots, \alpha_l\rangle |_x,$$ for certain locally defined $a$-forms $\alpha_1, \dots, \alpha_l.$ Then, an arbitrary form taking values in $S^a$ can be expressed as $$\alpha = f^i \alpha_i.$$
The condition of being closed translates into 
$$ 0 = \d{\alpha} = \d f^i \wedge \alpha_i + f^i  \d \alpha_i,$$
for certain functions $f^i \in C^\infty(M).$
This defines a set of partial differential equations on the coefficients $f^i$. It is clear that $f^i = 0$
defines the trivial closed form. In general, there may not exist non-zero solutions.

\begin{example} Let $M := \mathbb{R}^4$ with coordinates $(x, y ,z, t)$ and define the following vector subbundle
    $$S_2 := \langle (\d{x} + y \d{z}) \wedge \d{t} \rangle \subseteq \bigwedge^2 M.$$
    Any form $$\alpha: M \rightarrow S_2,$$ has the following expression
    $$f(\d{x} + y \d{z}) \wedge \d{t}, $$ for certain function $f(x, y , z , t).$ Then,
    \begin{align*}
        \d{\alpha} & = -\pdv{f}{y} \d{x} \wedge \d{y} \wedge \d{t}\\
        &+ \left( y \pdv{f}{x} - \pdv{f}{z} \right) \d{x} \wedge \d{z} \wedge \d{t}\\
        &+ \left( f + y \pdv{f}{y} \right) \d{y} \wedge \d{z} \wedge \d{t}.
    \end{align*}
    If $\alpha$ is closed, we have $f = 0$, and thus $\alpha = 0,$ showing that the only possible closed form taking values
    in $S^a$ is the trivial one.
\end{example}

This leads us to introduce the following definition:

\begin{Def} A vector subbundle $S^a \subseteq \bigwedge^a M$ is called \textbf{integrable} if for every $x \in M,$
    $\alpha_0 \in (S^a) \big |_x$ there exists a form (possible locally defined around $x$), $\alpha$, taking values in $S^a$ such that
    $\alpha |_x = \alpha_0$, and $\d{\alpha} = 0$.
\end{Def}

Essentially, the definition above guarantees the existence of enough locally Hamiltonian forms.

\begin{Def}[Poisson bracket] 
\label{Graded_Poisson_strcutures:def:bracket}
Let $S^a \subseteq \bigwedge^a M$, $1 \leq a \leq k$ be a sequence of integrable subbundles satisfying
    $$(S^k)^{\circ, 1} = 0,\,\, (S^a)^{\circ, p} = (S^b)^{\circ, p},$$ for each $p \leq a, b$. Denote by
    $$\Omega^{a-1}_H(U) := \langle  \alpha \in \Omega^{a-1}(U) : \d \alpha \in S^a\rangle $$ the set of Hamiltonian forms
    defined on certain open subset $U \subseteq M.$ A Poisson bracket of order $k$ on $M$ is a collection of bilinear operations
    $$\Omega^{a-1}_H(U) \otimes \Omega^{b-1}_H(U) \xrightarrow{\{\cdot, \cdot\}_U} \Omega^{a + b - (k+1)}_H(U),$$
    for each open subset $U \subseteq M$ satisfying the following properties. 
    For $\alpha \in \Omega^{a-1}_H(U), \beta \in \Omega^{b -1 }_H(U), \gamma\in \Omega^{c - 1}_H(U)$
    \begin{enumerate}[i)]
        \item \textit{It is graded:} $$\deg \{\alpha, \beta\}_U = \deg \alpha + \deg \beta;$$
        \item \textit{It is graded-skew-symmetric:} $$\{\alpha, \beta\}_U = -(-1)^{\deg \alpha \deg \beta} \{\beta,\alpha\}_U;$$
        \item \textit{It is local:} If $\d{\alpha} |_x = 0,$ $\{\alpha, \beta\}_U|_x = 0$
        \item \textit{It satisfies Leibniz identity:} Let $\beta^j \in \Omega^{b-1}_H(M), \gamma_j \in \Omega^{c-1}_H(M).$ If $\beta^j \wedge \d{\gamma} \in \Omega^{b+c - 1}_H(M),$ then, for $a = k,$ 
        $$\{ \beta^j \wedge \d{\gamma_j}, \alpha\}_U = \{\beta^j, \alpha\}_U \wedge \d{\gamma}_j + (-1)^{k -\deg \beta^j} \d{\beta} \wedge \{\gamma_j, \alpha\}_U;$$
        \item \textit{It is invariant by symmetries:} If $X \in \mathfrak{X}(U)$ and $\pounds_X \alpha = 0,$ then 
        $\iota_X \alpha \in \Omega^{a-2}_H(U)$ and
        $$\{\iota_X \alpha, \beta\}_U=  (-1)^{\deg \beta} \iota_X \{\alpha, \beta^j\}_U;$$
        \item \textit{It satisfies graded Jacobi identity (up to an exact term):} 
        $$(-1)^{\deg \alpha \deg \gamma}\{ \{\alpha, \beta \}_U , \gamma\}_U + \text{cyclic terms} = \text{exact form}.$$
    \end{enumerate}
    Furthermore, it satisfies the \textit{compatibility condition}, that is, for $V \subseteq U$ two open subsets, $\{\cdot, \cdot\}_V$
    is the restriction of $\{\cdot, \cdot\}_U.$
\end{Def}

\begin{remark} One may wonder why we require this version of Leibniz identity. The not-so-evident reason is that this is a strictly stronger hypothesis than the Leibniz rule from \cref{Graded_Poisson_structures:thm:Properties_Poisson_Bracket}.
Indeed, if $\beta^j \wedge \d{\gamma_j} \in \Omega^{b +c -1}_H(M),$ there is no way of guaranteeing that the individual
terms of the sum are Hamiltonian. This is the version needed in order to recover the graded Poisson structure from 
the graded Poisson bracket.
\end{remark}

We are now ready to state and prove the main result of this section:

\begin{theorem}  
\label{Graded_Poisson_brackets:thm:Bracket_Implies_Structure}
Let $S^a \subseteq \bigwedge^a M$, $1 \leq a \leq k$ be a sequence of integrable vector subbundles satisfying
    $$(S^k)^{\circ, 1} = 0,\,\, (S^a)^{\circ, p} = (S^b)^{\circ, p},$$ for each $p \leq a, b$, and let $\{\cdot, \cdot\}_U,$ $U \subseteq M$ open, be a Poisson bracket on this sequence. Suppose that the family $S^a$ satisfies the following properties:
    \begin{enumerate}[i)] 
        \item Locally, there exists Hamiltonian forms $\gamma_{ij} \in \Omega^{b-2}_H(U),$ and functions $f^j_i$ such that
        $$S^b = \langle \d{f}^j_i \wedge \d{\gamma}_{ij}, i\rangle;$$
        \item For each $1 \leq a \leq k $, locally, there exists a family of Hamiltonian forms forms ${\gamma}^j,$ and a family of vector fields $X^j$ such that 
        $$S^a = \langle \d {\gamma}^j\rangle \,\,\pounds_{X^j} \gamma^j = 0,$$ and $$S^{a-1} = \langle \d{\iota_{X^j} \gamma^j}\rangle.$$
    \end{enumerate}
    Then, there exists an unique graded Poisson structure $(S^a, K_p, \sharp_a)$ such that $\{\cdot, \cdot\}_U$ is the induced bracket by this
    structure. 
\end{theorem}

\begin{proof}We divide the proof in four steps.

    \begin{center}
        \textit{\underline{1. Definition of $(S^a, K_p, \sharp_a)$}}
    \end{center}
    Define 
    $$K_p\big |_x := (S^p\big |_x)^{\circ, p}.$$
    By hypothesis, these vector spaces define a vector subbundle of $\bigvee_p M$.
    In order to define $\sharp_a,$ notice that for each Hamiltonian form 
    $\alpha \in \Omega^{a-1}_H(U),$ the bracket induces an operation 
    $$\Omega^{k-a}_H(U) \xrightarrow{\Phi_\alpha} C^\infty(M), \, \beta \mapsto \Phi_\alpha(\beta) = (-1)^{\deg \alpha}\{\beta, \alpha\}.$$
    By locality, it factorizes through certain linear mapping (which only depends on $\d{\alpha}$) 
    $$\Psi_\alpha: S^{k +1 - a} \big |_U \rightarrow \mathbb{R}$$ as $$\Phi_\alpha(\beta) = \Psi_\alpha(\d \beta).$$
    Therefore, we may identify $\Psi_\alpha$ with an element of the dual $(S^{k+1-a})^\ast \big |_U,$ and we get a linear mapping
    $$\Psi : \Omega^{a-1}_H(U) \rightarrow (S^{k+1-a})^\ast \big |_U.$$ Since this mapping only depends on the exterior differential of $\alpha,$
    it again factorizes through
    $$\sharp_a : S^a \big |_U \rightarrow (S^{k+1-a})^\ast \big|_U = \left(\bigvee_{k+1-a}M / K_{k+1-a}\right) |_U,$$
    as $$\Psi_\alpha = \sharp_a(\d{\alpha}).$$
    By construction, this mapping satisfies 
    $$\{\beta, \alpha\}_U = (-1)^{\deg \alpha} \iota_{\sharp_a(\d \alpha)} \d \beta,$$ for $\alpha \in \Omega^{a-1}_H(U)$, and
    $\beta \in \Omega^{k-a}_H(U).$ It is clear that these mappings do not depend on the choice of open subset $U$. Therefore, we get globally well defined mappings $$\sharp_a: S^a \rightarrow \bigvee_{k+1-a} M/{K_{k+1-a}}$$ satisfying the equality above. From now on, making abuse of notation, we omit dependence on $U$.\\

    Steps 2 and 3 are devoted to proving that we have the equality $$\{\beta, \alpha\} = (-1)^{\deg \alpha} \iota_{\sharp_a(\d{\alpha})} \d{\beta}$$ for arbitrary Hamiltonian forms. We will prove it using two induction processes, summarized in the diagram below. Here, we write the case $k = 5$, and the order in which we obtain the equality in the orders $(\operatorname{ord} \beta, \operatorname{ord} \alpha).$ The lower triangle is omitted for degree considerations (the bracket is trivial when evaluated on pairs such that $\operatorname{ord} \beta + \operatorname{ord} \alpha < k - 1$).

\[\begin{tikzcd}[sep=large]
	{(0,4)} & {(1,4)} & {(2,4)} & {(3,4)} & {(4,4)} \\
	& {(1,3)} & {(2,3)} & {(3,3)} & {(4, 3)} \\
	&& {(2, 2)} & {(3,2)} & {(4,2)} \\
	&&& {(3,1)} & {(4,1)} \\
	&&&& {(4,0)}
	\arrow["{\textit{\underline{Step 2}}}", from=1-1, to=1-2]
	\arrow["{\textit{\underline{Step 2}}}", from=1-2, to=1-3]
	\arrow["{\textit{\underline{Step 3}}}", from=1-2, to=2-2]
	\arrow["{\textit{\underline{Step 2}}}", from=1-3, to=1-4]
	\arrow["{\textit{\underline{Step 3}}}", from=1-3, to=2-3]
	\arrow["{\textit{\underline{Step 2}}}", from=1-4, to=1-5]
	\arrow["{\textit{\underline{Step 3}}}", from=1-4, to=2-4]
	\arrow["{\textit{\underline{Step 3}}}", from=1-5, to=2-5]
	\arrow["{\textit{\underline{Step 3}}}", from=2-3, to=3-3]
	\arrow["{\textit{\underline{Step 3}}}", from=2-4, to=3-4]
	\arrow["{\textit{\underline{Step 3}}}", from=2-5, to=3-5]
	\arrow["{\textit{\underline{Step 3}}}", from=3-4, to=4-4]
	\arrow["{\textit{\underline{Step 3}}}", from=3-5, to=4-5]
	\arrow["{\textit{\underline{Step 3}}}", from=4-5, to=5-5]
\end{tikzcd}\]

    \begin{center}
        \textit{\underline{2. $\{\beta, \alpha\} = \iota_{\sharp_{k}(\d{\alpha})} \d{\beta}$, for an
        arbitrary Hamiltonian form $\beta$, and $\alpha \in \Omega^{k-1}_H(U)$.}}
    \end{center}
    We prove it by induction on the order of $\beta$. For $\beta \in C^\infty_H(M) = C^\infty(M),$ it is clear by definition
    of $\sharp_k.$ For the inductive step, let $\beta \in \Omega^{b-1}_H(M)$.\\

    Using the first required property on $S^a$, we only need to prove the result for a Hamiltonian form with the expression
    $$f^j \d{\gamma_{j}},$$
    where $\gamma_j$ is Hamiltonian. Indeed, by Leibniz identity:
    \begin{align*}
        \{f^j \d{\gamma}_j, \alpha\} &= \{f^j, \alpha\} \wedge \d{\gamma}_j - \d{f}^j \wedge \{\gamma_j, \alpha\}\\
        &= \left( \iota_{\sharp_k(\d{\alpha})} \d{f}^j\right)\wedge \d{\gamma}_j- \d{f}^j \wedge (\iota_{\sharp_k(\d{\alpha})} \d{\gamma_j})\\
        &= \iota_{\sharp_{k}(\d{\alpha})} \d{(f^j \d{\gamma_j})},
    \end{align*}
    where in the second equality we have used the induction hypothesis, proving that $$\{\beta, \alpha\} = \iota_{\sharp_a(\d{\alpha})} \d{\beta},$$ for $\alpha \in \Omega^{k-1}_H(M),$
    and $\beta \in \Omega^{b-1}_H(M).$

    \begin{center}
        \textit{\underline{3. $\{\beta, \alpha\} = (-1)^{\deg \alpha} \iota_{\sharp_{k}(\d{\alpha})} \d{\beta}$, for
        arbitrary Hamiltonian forms $\alpha$ and $\beta$.}}
    \end{center}

    Fix $\beta \in \Omega^{b-1}_H(M)$. We proceed by ``reverse'' induction on the order of $\alpha,$ $\alpha \in \Omega^{k-1}_H(M)$ being the base case from \textit{\underline{Step 2}}.
    \\

    Using the second required property on $S^a$, we only need to prove it for Hamiltonian forms $$\iota_X \gamma,$$ for certain vector field $X$ and certain Hamiltonian form $\gamma$ such that $\pounds_X \gamma = 0.$ Indeed, using invariance by symmetries and induction
    hypothesis:
    \begin{align*}
        \{\beta, \iota_X \alpha\} &= - (-1)^{\deg \iota_X \alpha \deg \beta}\{\iota_X \alpha, \beta\} = - (-1)^{(\deg \alpha +1) \deg \beta + \deg \beta} \iota_X\{\alpha, \beta\}\\
        &= - (-1)^{\deg\alpha \deg \beta} \iota_X \{\alpha, \beta\} = \iota_X\{\beta, \alpha\}\\
        &= (-1)^{\deg \alpha} \iota_X \iota_{\sharp_a(\d{\alpha})} \d{\beta} = (-1)^{\deg \alpha } \iota_{\sharp_a(\d{\alpha}) \wedge X} \d{\beta} 
    \end{align*}
    Now we wish to relate $\sharp_a(\d{\alpha}) \wedge X$ and $\sharp_{a-1}(\d{\iota_X \alpha}).$ This relation is given by the following Lemma:
    \begin{lemma} We have $\sharp_{a-1}(\d{\iota_X {\alpha}}) = -\sharp_a(\d{\alpha}) \wedge X + K_{k+2-a}.$
    \end{lemma}
    \begin{proof}
        Let $\gamma \in \Omega^{k - a}_H(M).$ On the one hand, by definition of $\sharp_{a-1},$ we have
        $$\{\gamma, \iota_X \alpha\} = (-1)^{\deg \alpha +1} \iota_{\sharp_{a-1}(\d{\iota_X \alpha})} \d{\gamma}.$$
        On the other hand, by invariance by symmetries, we have
        $$\{\gamma, \iota_X \alpha\} = \iota_X \{\gamma, \alpha\} =  (-1)^{\deg \alpha} \iota_X\iota_{\sharp_a(\d{\alpha})} \d{\gamma} =  (-1)^{\deg \alpha} \iota_{\sharp_a(\d{\alpha}) \wedge X} \d{\gamma}.$$
        Equating both of these equalities and using that $\gamma$ is arbitrary, we get the equality 
        $$\sharp_{a-1}(\d{\iota_X {\alpha}}) = -\sharp_a(\d{\alpha}) \wedge X + K_{k+2-a},$$
        proving the Lemma.
    \end{proof}
    Using the Lemma above, the result follows easily. Indeed,
    \begin{align*}
        \{\beta, \iota_X \alpha\} &= (-1)^{\deg \alpha} \iota_{\sharp_a(\d{\alpha}) \wedge X} \d{\beta}\\
        &= (-1)^{\deg \alpha + 1}\iota_{\sharp_{a-1}(\d{\iota_X \alpha})} \d{\beta} \\
        &= (-1)^{\deg \iota_X \alpha} \iota_{\sharp_{a-1}(\d{\iota_X \alpha})} \d{\beta}.
    \end{align*}

    Now it only remains to show that $(S^a, K_{k+1-a}, \sharp_a)$ is a graded Poisson structure.
    
    \begin{center}
        \textit{\underline{4. $(S^a, K_{k+1-a}, \sharp_a)$ defines a graded Poisson structure.}} 
    \end{center}

    Skew-symmetry is clear, using skew-symmetry of the brackets. To prove integrability, by \cref{Graded_Poisson_structures:thm:Equivalence},
     we only need to check it for $a = k.$ By the graded Jacobi identity, we have
    $$\{f, \{\alpha, \beta \}\} + \{\beta, \{f, \alpha\}\} + \{\alpha, \{\beta, f\}\} = 0,$$
    for every $\alpha, \beta \in \Omega^{k-1}_H(M)$, $f \in C^\infty(M)$
    This yields
    \begin{align*}
        \sharp_{k}(\d{\{\alpha, \beta\}}) (f) - \sharp_k(\d{\beta})(\sharp_k({\d{\alpha}})(f)) 
         + \sharp_k(\d{\alpha})(\sharp_k({\d{\beta}})(f)) = 0.
    \end{align*}
    Since $f \in C^\infty(M)$ is arbitrary, we have
    $$\sharp_{k}(\d{\{\alpha, \beta\}}) = - [\sharp_k(\d{\alpha}), \sharp_k(\d{\beta})],$$
    or, equivalently,
    $$\sharp_k \left( \pounds_{\sharp(\d{\beta})} \d{\alpha} \right) = - [\sharp_k(\d{\alpha}), \sharp_k(\d{\beta})].$$
    This implies integrability. Indeed, given arbitrary forms $\gamma_1, \gamma_2$ taking values in $S_k,$ we
    can express them locally as
    $$\gamma_1 = f^i \d{\alpha_i}, \,\, \gamma_2 = g^j \d{\beta_j},$$
    for certain Hamiltonian forms $\alpha_i, \beta_j \in \Omega^{k-1}_H(M)$, and functions $f^i, g^j \in C^\infty(M)$. Then,
    defining 
    \begin{align*}
        \theta &:= \pounds_{\sharp_k(\d{\gamma_1})} \gamma_2 - \iota_{\sharp_k(\d{\gamma_2})} \d{\gamma_1}\\
        &= \pounds_{f^j \sharp_k(\d{\alpha_j})} g^i \d{\beta_i} - g^i \iota_{\sharp_k(\d{\beta_i})} \d{(f^j \d{\alpha_j})}\\
        &= f^j \sharp_k(\d{\alpha_j})(g^i) \d{\beta_i} - g^i \sharp_k(\d{\beta_i})(f^j) \d{\alpha_j} + 
        f^ig^j\d{\iota}_{\sharp_k(\d{\alpha_j})} \d{\beta_i}\\
        &= f^j \sharp_k(\d{\alpha_j})(g^i) \d{\beta_i} - g^i \sharp_k(\d{\beta_i})(f^j) \d{\alpha_j} + 
        f^ig^j\d{\{\beta_i, \alpha_j\}},
    \end{align*}
    we have $\theta \in S_k$, and
    \begin{align*}
        \sharp_k(\theta) &=  f^j \sharp_k(\d{\alpha_j})(g^i) \sharp_k(\d{\beta^i}) - g^i \sharp_k(\d{\beta_i})(f^j)\sharp_k(\d{\alpha_j}) + 
        f^ig^j\sharp_k(\d{\{\beta, \alpha\}})\\
        &= f^j \sharp_k(\d{\alpha_j})(g^i) \sharp_k(\d{\beta_i}) - g^i \sharp_k(\d{\beta_i})(f^j)\sharp_k(\d{\alpha_j}) - 
        f^ig^j [\sharp_k(\d{\beta_i}), \sharp_k(\d{\alpha_j})]\\
        &= [f^j \sharp_{k}(\d{\alpha}_j), g^i \sharp_k(\d{\beta_i})] = [\sharp_k(\gamma_1), \sharp_k(\gamma_2)],
    \end{align*}
    proving that the graded Poisson structure is integrable.
\end{proof}

\begin{remark}
\cref{Graded_Poisson_brackets:thm:Bracket_Implies_Structure} generalizes the classical result that a Poisson bracket defined on $C^\infty(M)$ is characterized by an integrable bivector field, $\Lambda \in \mathfrak{X}^2(M),$ $[\Lambda, \Lambda] = 0.$ Indeed, since $(S_1)^{\circ, 1} = 0,$ we must have $S_1 = T^\ast M,$ which trivially satisfies the hypothesis.
\end{remark}

There are certain interesting cases that fulfill the hypothesis of
\cref{Graded_Poisson_brackets:thm:Bracket_Implies_Structure}:

\begin{corollary} 
\label{Graded_Poisson_structures:corollary:constant_coefficients}
If the family $S^a$, $1 \leq a \leq k$ is generated by forms of constant coefficients\footnote{Here, constant coefficients means that in certain local chart, the sequence $S^a$ (equivalently, $S^k$) is generated by forms of the type $\omega = c_{i_1, \dots, i_{a}}\d{x}^{i_1} \wedge \cdots \wedge \d{x}^{i_a},$ where $c_{i_1, \dots, i_a} \in \mathbb{R}$ are constants.}, any graded Poisson bracket defined on the space of Hamiltonian forms induces an unique graded Poisson structure.
\end{corollary}

\begin{proof} Indeed, in order to be in the hypotheses of \cref{Graded_Poisson_brackets:thm:Bracket_Implies_Structure},
we only need to check that the two required properties hold. For the first one, given certain local chart such that
$$S^a = \langle c^m_{i_1, \dots, i_a} \d{x^{i_1}} \wedge \cdots \wedge \d{x}^{i_a}, m \rangle,$$ with $c^m_{i_1, \dots, i_a} \in \mathbb{R}$, given that we can write
\begin{align*}
    c^m_{i_1, \dots, i_a} \d{x^{i_1}} \wedge \cdots \wedge \d{x}^{i_a} = \d{x}^{i_1} \wedge 
    \d{\left( c^m_{i_1, \dots, i_a} x^{i_2} \d{x}^{i_3}\wedge \cdots \wedge \d{x}^{i_a}\right)},
\end{align*}
it is enough to define
$$f^{mi} := x^i, \gamma^m_i := c^m_{i, i_2, \dots, i_a} x^{i_2} \d{x}^{i_3}\wedge \cdots \wedge \d{x}^{i_a},$$ and we have
$$S^a = \langle d f^{mi} \wedge \gamma^m_i, m\rangle.$$ 
For the second property, notice that, by the linear study, we know that
$$S^{a-1} = \langle \iota_{\pdv{x^j}} \omega^m, m\rangle,$$ where $\omega^m = c^m_{i_1, \dots, i_a} \d{x^{i_1}} \wedge \cdots \wedge \d{x}^{i_a}.$ Therefore, it is enough to prove that for each $j,$ there exists a primitive $(a-1)-$form $\gamma^{m}_j$ of $\omega^m$, $$\omega^m = \d{\gamma^m_j}$$ such that $\pounds_{\pdv{x^j}} \gamma^m_j = 0.$ Indeed, it suffices to define (here we omit Einstein summation convention)
\begin{align*}
    \gamma^m_j := &\sum_{\substack{i_1 < \cdots < i_k \\ i_1 \neq j \\ i_l = j}} c^m_{i_1 \cdots i_m} x^{i_1} \d{x^{i_2}} \wedge \cdots \wedge \d{x^{i_a}}\\
    &-  \sum_{j < i_2 \cdots < i_a} c^{m}_{ji_2 \dots i_a} x^{i_2} \d{x^j} \wedge \d{x ^{i_3}} \wedge \cdots \wedge \d{x^{i_a}}.
\end{align*}
We clearly have $\omega^m = \d{\gamma^m_j},$ and\footnote{Indeed, because in general, $\pounds_{\pdv{x^j}} f_{i_1 \cdots i_a} \d{x}^{i_1} \wedge \cdots \wedge \d{x} ^{i_a} = \pdv{f_{i_1 \dots i_a}}{x^j} \d{x^{i_1}} \wedge \cdots \wedge \d{x}^{i_a}$, and in our case the coefficients do not depend on $x^j$.} $\pounds_{\pdv{x^j}} \gamma^m_j = 0,$
finishing the proof.
\end{proof}

\begin{example}
    Recall (\cref{Overview:subsection:multisymplectic}) that the Hamiltonian formalism in classical field theories occurs in
    $$M := \bigwedge^n_2 Y  = \{ \alpha \in \bigwedge^n Y, \iota_{e_1 \wedge e_2} \alpha = 0, e_1, e_2 \in \ker d \pi\},$$
    where $$\pi: Y \rightarrow X$$ is a fibered manifold, and $n = \dim X.$ $\bigwedge^n_2 Y$ can be endowed with a non-degenerate multisymplectic form $\Omega$ which, in canonical coordinates $(x^\mu, y^i, p^\mu_i, p)$, takes the expression
    $$\Omega = \d{p} \wedge \d{^n} x + \d{p^\mu_i} \wedge \d{y^i} \wedge \d{^{n-1}} x_\mu.$$ Notice that it is of constant linear type.\\
    
    Now, $\Omega$ is a non-degenerate $(n+1)-$form and we can define the Poisson tensor 
    $$\sharp: S^k \rightarrow TM$$ as the inverse of 
    $$TM \xrightarrow{\flat} \bigwedge^k M, \, v \mapsto \iota_v \Omega.$$ More, precisely, we define
    $$S^k := \iota_{TM} \Omega,$$ and
    $$\sharp: S^k \rightarrow TM$$ by $\sharp(\alpha) := v$,  where $v$ is the unique vector satisfying $$\iota_v \Omega = \alpha.$$  In particular, being of constant linear type, 
    $$S^a := \iota_{\bigvee_{k+1-a} M} \Omega$$ defines a vector subbundle of $\bigwedge^a M$, and we get a well defined graded Poisson structure on $M$, $(S^a, K_{k+1-a}, \sharp_a),$ where $\sharp_a, K_p$ are the mappings and subbundles obtained from \cref{Graded_Poisson_structures:thm:Equivalence}. In terms of the multisymplectic form, $\Omega,$ we have $$K_p = \Omega^{\circ, p},$$ and $$\sharp_a( \alpha) = U + K_p$$ if and only if $$\iota_U \Omega = \alpha.$$ Since the family $S^a$ is generated by forms of constant coefficients, we are in the hypotheses of \cref{Graded_Poisson_structures:corollary:constant_coefficients}, and we conclude that the induced graded Poisson bracket completely determines the multisymplectic structure on $M = \bigwedge^n_2 Y.$ 
\end{example}

We can also obtain the multisymplectic foliation from \cref{Graded_Poisson_structures:thm:Poisson_implies_Multisymplecticfoliation} 
using the Poisson bracket. For certain $\alpha \in \Omega^{k-1}_H(M)$ (which may be locally defined), since 
$$C^\infty(M) \rightarrow C^\infty(M); \, f \mapsto \{f, \alpha\}$$ is a derivation by locality, there exists an unique vector field,
which we shall denote $X_\alpha$ such that $$\{f, \alpha\} = X_\alpha(f),$$
for every $f \in C^\infty(M).$ Of course, in tensorial terms, $$X_\alpha = \sharp_k(\d{\alpha}).$$

\begin{Def} [Hamiltonian vector field] A \textbf{Hamiltonian vector field} of a graded Poisson structure $(S^a, K_p, \sharp_a)$ is a vector field
    of the form $X_\alpha,$ for certain Hamiltonian $(k-1)-$form $\alpha$. 
\end{Def}

Then, we clearly have

\begin{proposition} The characteristic distribution $E = \Im \sharp_k$ of the graded Poisson structure is generated by (locally defined) 
    Hamiltonian vector fields, that is, 
    $$E = \langle  X_\alpha: \alpha \in \Omega^{k-1}_H(M)\rangle.$$
\end{proposition}

Under this description, integrability is clear. Indeed, using Jacobi identity, we have that the mapping 
$$\Omega^{k-1}_H(M) \rightarrow \mathfrak{X}(M), \, \alpha \mapsto X_\alpha$$ defines a Lie anti-homomorphism,
that is,  $$[X_\alpha, X_\beta] = -X_{\{\alpha, \beta\}},$$ from which involutivity follows. Now we can describe 
the multisymplectic form in each leaf $F \in \mathcal{F}$ of the foliation as follows:

\begin{proposition} Let $\mathcal{F}$ be the foliation induced by the (integrable) characteristic distribution. Then, the
    multisymplectic form $\omega_F$ defined in \cref{Graded_Poisson_structures:thm:Poisson_implies_Multisymplecticfoliation}
    can be expressed as 
    $$(\omega_F|_x )(X_{\alpha_1} |_x, \dots, X_{\alpha_{k+1}}|_x) := (\d{\alpha_1}|_x)(X_{\alpha_2}|_x, \dots, X_{\alpha_k}|_x),$$
    for every $x \in F$.
\end{proposition}
\begin{proof} It follows from the description of $\omega_F$ given in 
    \cref{Graded_Poisson_structures:thm:Multisymplecticfoliation_implies_Poisson}.
\end{proof}

\begin{remark} Similarly to how both \cref{Graded_Poisson_structures:thm:Poisson_implies_Multisymplecticfoliation} and \cref{Graded_Poisson_structures:thm:Multisymplecticfoliation_implies_Poisson}  can be generalized to graded Dirac structures, we get the corresponding generalization of \cref{Graded_Poisson_brackets:thm:Bracket_Implies_Structure} to graded Dirac structures. First, notice that the concept of Hamiltonian form, and that of the Poisson bracket is easily generalized using the description given in \cref{Dirac_Structures:Linear_Dirac_Structures:prop:EquivalentDescription_GradedDirac}. In terms of the vector
subbundles $D_1, \dots, D_k$, the space of $(a-1)$-Hamiltonian forms is 
$$\Omega^{a-1}_H(M) = \{ \alpha \in \Omega^{a-1}(M): (U,\d{\alpha}) \in D_{k+1 - a}, \text{ for certain } U \in \bigvee_{k+1-a}M\},$$ and the bracket of two Hamiltonian forms is
$$\{\alpha, \beta\} := (-1)^{\deg \beta} \iota_V \d{\alpha},$$ where $(U, \d{\alpha}) \in D_{k+1-a},$ and $(V, \d{\beta}) \in D_{k+1-b}.$ The condition of the Poisson bracket of two Hamiltonian being again Hamiltonian follows from involutivity with respect to the Courant bracket, since 
$$\llbracket (U, \alpha), (V, \beta)\rrbracket = ([U,V], - \d\{\alpha, \beta\}) \in D_{2k + 2 - (a + b)},$$
as a quick calculation shows.
Furthermore, the properties of the corresponding Poisson bracket are exactly the same as the obtained in \cref{Graded_Poisson_structures:thm:Properties_Poisson_Bracket}. Dropping the non-degeneracy condition $(S^k)^{\circ, 1} = 0$, we can define a graded Poisson bracket as a collection of brackets $$\Omega^{a-1}_H(U) \otimes \Omega^{b-1}_H(U) \xrightarrow{\{\cdot, \cdot \}_U} \Omega^{a + b - (k+1)}$$ satisfying all the properties listed in \cref{Graded_Poisson_strcutures:def:bracket}. Then, we obtain the corresponding result in graded Dirac manifolds:
\end{remark}

\begin{theorem} \label{GradedPoissonStructures:thm:Bracket_implies_GradedDirac}
Let $S^a \subseteq \bigwedge^a M$, $1 \leq a \leq k$ be a sequence of integrable vector subbundles satisfying
    $$ (S^a)^{\circ, p} = (S^b)^{\circ, p},$$ for each $p \leq a, b$, and let $\{\cdot, \cdot\}_U,$ (where $U \subseteq M$ is an open set), be a Poisson bracket on this sequence. Suppose that the family $S^a$ satisfies the following properties:
    \begin{enumerate}[i)]
        \item Locally, there exists Hamiltonian forms $\gamma_{ij} \in \Omega^{b-2}_H(U),$ and functions $f^j_i$ such that
        $$S^b = \langle \d{f}^j_i \wedge \d{\gamma}_{ij}, i\rangle;$$
        \item For each $1 \leq a \leq k $, locally, there exists a family of Hamiltonian forms forms ${\gamma}^j,$ and a family of vector fields $X_i$ such that 
        $$TM = \langle X_i \rangle, \,\,S^a = \langle \d {\gamma}^j\rangle,$$ and $$\pounds_{X_i} \gamma^j = 0.$$
    \end{enumerate}
    Then, there exists an unique graded Dirac structure $D_1, \dots, D_k$ such that $\{\cdot, \cdot\}_U$ is the induced bracket by this
    structure. 
\end{theorem}
\begin{proof} The same argument from \cref{Graded_Poisson_brackets:thm:Bracket_Implies_Structure} gives us a family $(S^a, K_{k+1-a}, \sharp_a)$ satisfying the hypothesis of \cref{Dirac_Structures:Linear_Dirac_Structures:prop:EquivalentDescription_GradedDirac} and, therefore, it determines an unique graded Dirac structure, $D_1, \dots, D_k$.
\end{proof}

\begin{remark} This result clearly resembles the original idea behind Dirac structures, structures arising from Poisson algebras defined on a suitable subclass of functions on a manifold (see \cite{Courant1986}). In this case, we get a graded Dirac structure arising from a graded Poisson bracket defined on a subclass of forms.
\end{remark}


\subsection{Currents and conserved quantities}
\label{Graded_Poisson_structures:subsection:Currents}

Recently, F. Gay-Balmaz, J. C. Marrero, and N. Martínez-Alba introduced a bracket formulation of classical field theory in \cite{GayBalmaz2024}. Following the notation of \cref{ex:Classical_field_theory}, a Hamiltonian in classical field theory is a section
$$h: \bigwedge^n_2 Y / \bigwedge^n_1 Y \rightarrow \bigwedge^n Y_2,$$ where $$\bigwedge^n_1 Y = \{\text{semi-basic $n$-forms over }X\} = \{\alpha \in \bigwedge^n Y: \iota_e \alpha = 0, \text{when } e \in \ker \d{\pi}\}.$$ The Hamiltonian section has the local expression
$$h(x^\mu, y^i, p^\mu_i) = (x^\mu, y^i, p^\mu_i, p = -H),$$ 
where, usually, $$H = \pdv{L}{z^i_\mu} z^\mu_i - L$$ is the Hamiltonian function associated to a regular Lagrangian. This bracket is defined for a very specific family of forms, those $\alpha \in \Omega^{n-1}(Z^\ast)$\footnote{where we are denoting $Z^\ast = \bigwedge^n_2 Y / \bigwedge^n_1 Y.$} with $$\alpha = \left(A^i p^\mu_i + B^\mu\right) \d{^{n-1}}x_\mu,$$ where $A^i, B^i$ only depend on $(x^\mu, y^i)$. Then, their bracket is a semi-basic $n$-form over $X$
$$\{\alpha, h\}: \bigwedge^n_2 Y \rightarrow \bigwedge^n X$$ which takes the following local expression
$$\{\alpha, h\} = \left(\pdv{A^i}{x^\mu} p^\mu_i + \pdv{H}{p^\mu_j} \left( \pdv{A^i}{y^j} p^\mu_i + \pdv{B^\mu}{y^j}\right) - \pdv{H}{y^i} A^i\right) \d{^n}x.$$

This bracket measures the evolution of the observable $\alpha$ on a solution of the system determined by $h$. We can define it intrinsically using the geometry of the Poisson bracket on $\bigwedge^n_2 Y$ and, furthermore, we can extend the domain of definition. The property that this bracket satisfies follows from the properties of $\sharp_{n}$, as we will shortly see.\\

First, we notice that the Hamiltonian section is completely determined by the following $n$-form:

\begin{Def} Given a Hamiltonian section $h$, define the form $$\widetilde h: \bigwedge^n_2 Y \rightarrow \bigwedge^n X$$ as $\widetilde h(\alpha) := \alpha - h(\alpha + \bigwedge^n_1 Y).$
\end{Def}

Locally, 
$$\widetilde h(\alpha) = (p + H) \d{^n}x.$$

\begin{proposition} If $\alpha \in \Omega^{n-1}(Z^\ast)$ has the local expression 
$$\alpha = \left(A^i p^\mu_i + B^\mu\right) \d{^{n-1}}x_\mu,$$ where $A^i, B^i$ only depend on $(x^\mu, y^i),$ then $\tau^\ast \alpha$ is Hamiltonian, being $$\tau: \bigwedge^n_2 Y \rightarrow Z^\ast$$ the natural projection.
\end{proposition}
\begin{proof} Indeed, denoting $\widetilde \alpha := \tau^\ast \alpha,$ we have
\begin{align*}
    \d{\widetilde \alpha} =& \left(\pdv{A^i}{x^\mu}p^\mu_i + \pdv{B^\mu}{x^\mu} \right) \d{^n}x\\
    &+ \left(\pdv{A^i}{y^j}p^\mu_i + \pdv{B^\mu}{y^j} \right) \d{y^j} \wedge \d{^{n-1}}x_\mu\\
    &+ A^i \d{p^\mu_i} \wedge \d{^{n-1}}x_\mu.
\end{align*}
Then, it is clear that if we define 
$$X_\alpha := \left(\pdv{A^i}{x^\mu}p^\mu_i + \pdv{B^\mu}{x^\mu} \right) \pdv{p} + \left(\pdv{A^i}{y^j}p^\mu_i + \pdv{B^\mu}{y^j} \right) \pdv{p^\mu_j} - A^i \pdv{y^i},$$
we have 
$$\iota_{X_\alpha} \Omega = \d{\widetilde \alpha},$$
proving the result.
\end{proof}

The relation between the bracket defined in \cite{GayBalmaz2024} and our theory is the following
\begin{proposition} We have that $$\{\alpha, h\} = \iota_{X_\alpha}\d{\widetilde h}$$
\end{proposition}
\begin{proof} It is an immediate calculation.
\end{proof}

Therefore, we have 
$$\{\alpha, h\} = \iota_{\sharp_{n}(\d{\widetilde\alpha})} \d{\widetilde h},$$
providing the evolution of any observable in terms of the tensor $\sharp_n.$ We can use this description to extend the domain of definition:

\begin{Def}[Observable] We say that an $(n-1)-$form $\alpha \in \Omega^{n-1}(Z^\ast)$ is an \textbf{observable} if it induces a Hamiltonian form on $\bigwedge^n_2 Y$, that is, if $\tau^\ast \alpha \in \Omega^{n-1}_H\left(\bigwedge^n_2 Y\right),$ where $$\tau: \bigwedge^n_2 Y \rightarrow Z^\ast$$ denotes the natural projection. Denote by $\mathcal{O}$ the set of all observables.
\end{Def}
Also, notice that the form induced $\widetilde h$ by any Hamiltonian section $h$ is a semi-basic $n$-form over $X$ on $\bigwedge^n_2 Y$. Thus, we define the following bracket.

\begin{Def} Denote by {$\mathcal{A}$} the linear space of all semi-basic $n$-forms over $X$ defined on $\bigwedge^n_2 Y$ and define
$$ \mathcal{O} \otimes \mathcal{A}\xrightarrow{\{\cdot, \cdot\}} \mathcal{A}$$ by 
$$\{\alpha, \eta\} := \iota_{\sharp_n({ \d{\widetilde \alpha}})} \d{\eta},$$ where $\widetilde \alpha =\tau^\ast \alpha \in \Omega^{n-1}_H\left( \bigwedge^n_2 Y\right).$
\end{Def}

Of course, we have to check well-definedness

\begin{proposition}
\label{Graded_Poisson_structures:prop:Currents}
We have 
$$\{\alpha, \eta\} \in \mathcal{A},$$
for every $\alpha \in \mathcal{O}, \eta  \in \mathcal{A}.$
\end{proposition}
\begin{proof} Indeed, we prove it in coordinates. 
In general, if $\widetilde \alpha = \tau^\ast \alpha$ is Hamiltonian, there exists a vector field 
$$X = A \pdv{p} + B^\mu_i \pdv{p^\mu_i} + C^i\pdv{y^i} + D^\nu \pdv{x^\nu}$$ such that
$$\d{\widetilde \alpha} = \iota_X \Omega,$$ that is, such that
\begin{align*}
    \d{\widetilde \alpha} = &A \d{^n} x + B^\mu_i \d{y^i}\wedge \d{^{n-1}}x_\mu - C^i \d{p^\mu_i} \wedge \d{^{n-1}} x_\mu
    + D^\nu \d{p^\mu_i} \wedge \d{y^i} \wedge \d{^{n-2}} x_{\mu \nu}\\
    &- D^\nu \d{p} \wedge \d{^{n-1}}x_\nu.
\end{align*}
Now, since $\widetilde \alpha = \tau^\ast \alpha,$ $\d{\widetilde \alpha}$ must be zero when contracted with a $\tau-$vertical vector field (a multiple of $\pdv{p}$). Hence, $D^\nu = 0,$ and we have
$$
\sharp_{n}(\d{\widetilde \alpha}) = A \pdv{p} + B^\mu_i \pdv{p^\mu_i} + C^i\pdv{y^i}.
$$
Suppose $$\eta = f \d{^n}x,$$ where $f \in C^\infty\left( \bigwedge^n_2 Y\right).$
Since $$\d{\eta} = \d{f} \wedge \d{^n}x,$$ we have
\begin{align*}
    \{\alpha, \eta\} &= \iota_{\sharp_n(\d{ \alpha})} \d{\eta}\\
    &= \left(A \pdv{f}{p} + \pdv{f}{y^i} C^i + \pdv{f}{p^\mu_i} B^\mu_i\right) \d{^n x},
\end{align*}
which is semi-basic over $X$, concluding the proof.
\end{proof}

Notice that the Poisson bracket on $\Omega^{n-1}_H\left( \bigwedge^n_2 Y\right)$ induces a bracket on $\mathcal{O}.$ Indeed, if $\alpha, \beta \in \mathcal{O},$ and 
\begin{align*}
    \sharp_n(\d{\widetilde \alpha}) = A \pdv{p} + B^i_\mu \pdv{p^i_\mu} + C^i\pdv{y^i},\\
    \sharp_n(\d{\widetilde \beta}) = \widetilde{A} \pdv{p} + \widetilde{B}^i_\mu \pdv{p^i_\mu} + \widetilde{C^i}\pdv{y^i},
\end{align*}
where $\widetilde \alpha = \tau^\ast \alpha, \widetilde \beta = \tau^\ast \beta,$ we have\footnote{$\{\cdot, \cdot\}$ denotes the Poisson bracket of Hamiltonian $(n-1)-$forms.}
$$\{\widetilde \alpha, \widetilde \beta\} = \left(B^\mu_i \widetilde C^i - C^i \widetilde{B}^i_\mu\right)\d{^{n-1}}x_\mu,$$ which is clearly a semi-basic $(n-1)-$form on $\bigwedge^n_2 Y$ over $Z^\ast,$ so it induces an unique form, which we will denote\footnote{We are aware that this notation overlaps with the Courant bracket defined in \cref{Graded_Dirac_structures:Graded_Courant_bracket}; however, we hope that the reader will distinguish them based on the context they are used on.}
$$\llbracket \alpha, \beta\rrbracket \in \mathcal{O}.$$
\begin{remark} $(\mathcal{O}, \llbracket\cdot, \cdot \rrbracket)$ \textit{is not} a Lie algebra. Indeed, we know that $\{\cdot, \cdot\}$ induces a Lie algebra structure on $\Omega^{n-1}_H\left( \bigwedge^n_2 Y\right)$ modulo exact forms. This behaviour is transmitted onto $(\mathcal{O}, \llbracket\cdot, \cdot \rrbracket)$. Nevertheless, when we restrict this bracket to the space of forms with the local expression
$$\alpha = \left(A^i p^\mu_i + B^\mu\right) \d{^{n-1}}x_\mu,$$ where $A^i, B^i$ only depend on $(x^\mu, y^i),$ it does define a Lie algebra (see \cite{GayBalmaz2024, CastrillonLopez2003}).
\end{remark}

Now we have

\begin{theorem} Let $\alpha, \beta \in \mathcal{O},$ and $\eta \in \mathcal{A}$. Then, 
$$\{\llbracket \alpha, \beta\rrbracket, \eta\} = - \{\alpha, \{\beta, \eta\}\} + \{\beta, \{\alpha, \eta\}\}.$$
\end{theorem}
\begin{proof} Indeed, since both $\sharp_n(\d{\widetilde \alpha}), \sharp_n(\d{\widetilde \beta})$ are vertical vector fields on the bundle 
$$\bigwedge^n_2 Y \rightarrow Z^\ast,$$ we have that
\begin{align*}
    \{\llbracket \alpha, \beta\rrbracket, \eta\} &= \iota_{\sharp_n(\d{\{\widetilde \alpha, \widetilde \beta\})}} \d{\eta}  = - \iota_{[\sharp_n( \d{\widetilde \alpha}), \sharp_n(\d{\widetilde \beta})]} \d{\eta}\\
    &= -\pounds_{\sharp_n(\d{\widetilde \alpha})} \iota_{\sharp_n(\d{\widetilde \beta})} \d{\eta}+
    \iota_{\sharp_n(\d{\widetilde \beta})}\pounds_{\sharp_n(\d{\widetilde \alpha})}  \d{\eta}\\
    &= -\iota_{\sharp_n(\d{\widetilde \alpha})} \d{\iota_{\sharp_n(\d{\widetilde \beta})}} \d{\eta}
    +\iota_{\sharp_n(\d{\widetilde \beta})} \d{\iota_{\sharp_n(\d{\widetilde \alpha})}} \d{\eta}\\
    &=  - \{\alpha, \{\beta, \eta\}\} + \{\beta, \{\alpha, \eta\}\}.
\end{align*}
\end{proof}

\begin{corollary} When restricted to the forms with local expression $$\alpha = \left(A^i p^\mu_i + B^\mu\right) \d{^{n-1}}x_\mu,$$ where $A^i, B^i$ only depend on $(x^\mu, y^i),$ the mapping $$\alpha \mapsto \{\alpha, \cdot\}$$ defines a linear anti-representation on $\mathcal{A}$ of the Lie algebra induced by restricting $\llbracket \cdot, \cdot \rrbracket$
\end{corollary}

\begin{obs} Notice that the bracket 
$$\mathcal{O}\otimes \mathcal{A} \xrightarrow{\{\cdot, \cdot\}} \mathcal{A}$$ 
restricts well to semi-basic forms over $X$ which are also basic over $Z^\ast,$ that is, forms with the expression $$\alpha = f \d{^n}x,$$ where $f = f(x^\mu, y^i, p^\mu_i).$ Therefore, if we denote by $$\mathcal{B} := \{\text{semi basic $n$-forms over }X\} \cap \{\text{basic $n$-forms over } Z^\ast\},$$ we have an induced bracket 
$$\mathcal{O} \otimes \mathcal{B} \xrightarrow{\{\cdot, \cdot\}} \mathcal{B},$$ for which the previous corollary still holds. Also notice (and this is the main reason to consider the previous restriction) that the space of Hamiltonian sections $$\Gamma \left( \bigwedge^n_2 Y \rightarrow Z^\ast\right)$$ is an affine space modelled by $\mathcal{B}.$ Furthermore, if we define $$\{\alpha, h\} := \iota_{\sharp_n(\d{\alpha})} \d{\widetilde h}\footnote{
where $h(x^\mu, y^i, p^\mu_i) = (x^\mu, y^i, p^\mu_i, p = - H),$ $\widetilde h = (p + H) \d{^n x}$}, $$
we have $$\{\alpha, h\} \in \mathcal{B}.$$
\end{obs}

Therefore, we have an affine bracket 
$$\mathcal{O} \otimes \Gamma\left(\bigwedge^n_2 Y \rightarrow Z^\ast\right) \xrightarrow{\{\cdot, \cdot\}} \mathcal{B}$$ that satisfies the equality
$$\{\llbracket\alpha, \beta \rrbracket, h\} = -\{\alpha, \{\beta, h\}\} + \{\beta, \{\alpha, h\}\}.$$
This implies the following
\begin{corollary}[\cite{GayBalmaz2024}] When restricted to forms with the expression $$\alpha = \left(A^i p^\mu_i + B^\mu\right) \d{^{n-1}}x_\mu,$$ where $A^i, B^i$ only depend on $(x^\mu, y^i),$ the previous bracket defines an affine Lie algebra representation (for more details we refer to the cited paper).
\end{corollary}

Now that we have seen how to interpret (and extend the domain of definition) of this bracket previously introduced in the literature, let us prove that, even though we have extended the domain of definition, it still measures the evolution of observables.  

\begin{Def} A section $$\psi: X \rightarrow Z^\ast$$ is called a solution for the Hamilton-De Donder-Weyl equations defined by a Hamiltonian section 
$$h: Z^\ast \rightarrow \bigwedge^n_2 Y,$$ if $$\psi^\ast \iota_\xi \Omega_h = 0, \forall \xi \in \mathfrak{X}(Z^\ast),$$ where $\Omega_h = h^\ast \Omega.$
\end{Def}
Locally, this is equivalent to $\psi$ solving the following set of partial differential equations
\begin{align*}
\pdv{\psi^i}{x^\mu} = \pdv{H}{p^\mu_i},\\
\pdv{\psi^\mu_i}{x^\mu} = -\pdv{H}{y^i}.
\end{align*}
\begin{theorem} If $\psi$ is a solution for Hamilton-De Donder- Weyl equations, and $\alpha \in \Omega^{n-1}(Z^\ast)$ is an observable, we have
$$\psi^\ast (\d{\alpha}) = (h \circ \psi)^\ast \{\alpha, h\}.$$
\end{theorem}
\begin{proof} It is an immediate calculation in coordinates. Using the notation from \cref{Graded_Poisson_structures:prop:Currents}, we have
$$\d{\alpha} = A \d{^n} x + B^\mu_i \d{y^i} \wedge \d{^{n-1}}x_\mu - C^i \d{p^\mu_i} \wedge \d{^{n-1}} x_\mu,$$ and
$$\{\alpha, h\} = \left(A + \pdv{H}{y^i}C^i + \pdv{H}{p^\mu_i} B^\mu_i \right) \d{^n}x.$$
Therefore, 
\begin{align*}
    \psi^\ast(\d{\alpha}) = \left(A - \pdv{\psi^\mu_i}{x^\mu}C^i + \pdv{\psi^i}{x^\mu} B^\mu_i \right) \d{^n}x,
\end{align*}
and we obtain the result using the local version of Hamilton-De Donder-Weyl equations.
\end{proof}

As a consequence, we obtain the following:

\begin{theorem}If an observable $\alpha \in \Omega^{n-1}_H(M)$ satisfies $$\{\alpha, h\} = 0,$$ then it is conserved for any solution of the Hamilton-De Donder-Weyl equations defined by the Hamiltonian section $h$\footnote{where conserved means that $\psi^\ast (\d{\alpha}) = 0$}. 
\end{theorem}


\section{Conclusions and further work}
\label{section:Conclusions}
In this paper, we reviewed the different definitions of Poisson and Dirac structures of higher degree, introducing some new concepts which (under some mild regularity conditions) include all previous concepts present in the literature, generalizing Dirac geometry to a graded version. We also proved some results whose analogue in classical mechanics was of fundamental importance: the recovery of graded Poisson structures from graded Poisson brackets, and its connection to currents and conserved quantities (\cref{Graded_Poisson_brackets:thm:Bracket_Implies_Structure}, \cref{Graded_Poisson_structures:subsection:Currents}).\\


As a result of this work, we expect this theory to apply in the study of classical field theories:

\begin{enumerate}[i)]

\item The usual Dirac structures are not only a common framework for presymplectic and Poisson geometries, but allow us to include those pairs formed by a vector field and a 1-form on the manifold that are related by structure. Thus, in field theories, we can pair in different degrees multivector Hamiltonians and currents. We expect this theory to permit to discuss the so-called higher-form symmetries \cite{Gomes2023}, a hot topic in theoretical physics.

\item Although the graded nature of the objects studied in this paper (graded Poisson and graded Dirac structures) are characterized by the structure induced on vector fields and $k$-forms, we believe that the pairs of multivector fields and forms of different order will be of fundamental importance in the study of classical field theories. As an example, we have the recovery of the graded Poisson structure given a graded Poisson bracket (\cref{Graded_Poisson_brackets:thm:Bracket_Implies_Structure}), where the interplay between different degree objects seems essential.

\item In the near future, we plan on using the graded brackets to study the properties of the distinguished submanifolds of the evolution space of the theory.

\item  We also aim to extend the results of this paper to an analogous study for the case of classical action-dependent field theories (see \cite{Leon2023, Leon2024b}).

\item The study of singular Lagrangians to obtain a well-posed problem through a constraint algorithm may benefit from the use of graded Dirac structures, and so it could be an interesting direction for research. For applications of Dirac structures to the study of Dirac brackets, we refer to \cite{Bhaskara1988}. In this line of thought, an algebraic formalism like the one proposed in \cite{Ibort1999} for the brackets appearing in classical mechanics could also be extended to the theory presented in this paper.

\item Similar to the results in \cite{Grabowski1997},
it is important to understand the possible extensions of the bracket to families of forms where the order is not restricted to $0 \leq a \leq k-1,$ but can vary freely. Perhaps this could give an interpretation of the bracket given \cref{Graded_Poisson_structures:subsection:Currents} in terms of an extension of the graded Poisson bracket studied in this text.

\item Another possible direction for research is the extension of moment maps and reduction from symplectic geometry to this new setting. Previous work regarding these questions can be found in \cite{Callies2016, Fregier2015, Ryvkin2015, DELEON2004, Blacker2021}.
\end{enumerate}

\section*{Acknowledgments}

We acknowledge financial support of the Ministerio de Ciencia, Innovación y Universidades (Spain), grants PID2022-137909NB-C21 and RED2022-134301-T. We also acknowledge financial support from the Severo Ochoa Programmes for Centers of Excellence in R\&D (CEX2019-000904-S and CEX2023-001347-S). Rubén Izquierdo-López also acknowledges a Severo Ochoa-ICMAT scholarship for master students. Finally, we would like to thank the referee for the thorough revision of the paper and suggestions, which improved its quality. 

\appendix
\section{Schouten-Nijenhuis bracket}
The Schouten-Nijenhuis bracket is a generalization of the classical Lie bracket to multivector fields
$$\mathfrak{X}^p(M) \otimes \mathfrak{X}^q(M) \xrightarrow{[\cdot, \cdot]} \mathfrak{X}^{p+q-1}(M).$$
For locally decomposable multivector fields, it is defined as
\begin{align*}
    &[X_1 \wedge \cdots \wedge X_p, Y_1 \wedge \cdots \wedge Y_q] :=\\
    &\sum_{i = 1}^p \sum_{j= 1}^{q} (-1)^{i+j}[X_i, Y_j] \wedge X_1 \wedge \cdots \wedge \widehat{X}_i \wedge \cdots X_p \wedge Y_1 \wedge \cdots \widehat{Y}_j \wedge \cdots Y_q,
\end{align*}
and is extended to arbitrary multivector fields by linearity. We use the sign conventions of \cite{Marle1997}, \cite{Forger2003a}. Note that \cite{Vaisman1994} uses different sign conventions, differing by a factor of $(-1)^{p+1}$ from ours.\\

Define the Lie derivative of a form $\omega \in \Omega^a(M)$ with respect to a multivector field $U \in \mathfrak{X}^p(M)$ as $$\pounds_U \omega := \d{\iota}_U \omega - (-1)^p \iota_U \d{\omega}.$$ We have the following properties, which we enunciate without proving them (we refer to \cite{Forger2003a} for proofs).

\begin{theorem}[Properties of the Schouten-Nijenhuis bracket and the Lie derivative] Let $U, V, W$ be multivector fields of order $p, q , r$, respectively, and $\omega \in \Omega^a(M).$ Then,
\begin{enumerate}[i)]
    \item $$[U, V] = - (-1)^{(p-1)(q-1)} [V, U],$$
    \item $$\iota_{[U,V]} \omega =(-1)^{(p-1)q} \pounds_U \iota_V \omega - \iota_V \pounds_U \omega,$$
    \item $$[U, V \wedge W] = [U, V ] \wedge W + (-1)^{(p-1)q} V \wedge [U, W],$$
    \item $$(-1)^{(p-1)(q-1)} [U, [V, W]] + \text{cyclic} = 0.$$
\end{enumerate}
\end{theorem}


\phantomsection
\addcontentsline{toc}{section}{References}
\printbibliography
\end{document}